\documentclass[11pt]{article}

\usepackage[main]{preamble}

\begin{document}
\maketitle

\begin{abstract}
We give a proof of the cubic parallel-repetition bound for finite
two-player entangled games using quantum conditional relative entropy.
For a game with entangled value $1-\varepsilon<1$ and answer alphabets
$\mathcal A,\mathcal B$, its $n$-fold repeated value is at most
\[
 \exp[-\Omega(\varepsilon^3n/L)],\qquad
 L:=\max\{1,\log(|\mathcal A||\mathcal B|)\},
\]
for every $n\ge1$. This recovers Song's cubic rate.
We also prove a threshold bound: the probability of winning a fraction
at least $1-\varepsilon+\delta$ is at most
$\exp[-\Omega(\delta^5n/(\varepsilon^2L))]$ for
$0<\delta\le\varepsilon$.

Our information charge is the decrease in relative entropy when a
question reference is discarded. Our comparison states preserve the
original conditional question law. A chain rule bounds the coordinate
charges using the cost of conditioning on checked wins and retaining
checked answers. Our main lemma converts small charges into nearby
pure-state descriptions and local corrections, each depending only on
one player's question. Together with Song's sampler, these enable
approximate local preparation of the conditioned workspace.
\end{abstract}
\tableofcontents

\section{Introduction}
\label{sec:intro}
\label{sec:introduction}
Parallel repetition is a basic method for reducing the error of
multiprover interactive proofs. A two-player one-round game sends
classical questions to two separated players and checks their classical
answers. In an entangled game the players may share a quantum state
before the questions arrive, but cannot communicate afterward. The
repeated game draws question pairs independently across coordinates and
accepts only if every coordinate is won. Joint measurements across the coordinates can
correlate the answers, so exponential decay requires a proof.

Raz established parallel repetition for classical games \cite{Raz},
and Holenstein gave an information-theoretic proof with cubic dependence
on the gap from value one \cite{Holenstein}. Rao established concentration
above the classical value for general games \cite{RaoConcentration}.
For entangled games,
Cleve, Slofstra, Unger, and Upadhyay proved perfect repetition for XOR
games \cite{CSUU}; Kempe, Regev, and Toner treated unique games
\cite{KRT}; and Dinur, Steurer, and Vidick treated projection games
\cite{DSV}. Jain, Pereszl\'enyi, and Yao (JPY) proved repetition for
product question distributions, in which the two players' questions are
independent \cite{JPY}. For these free games, Chailloux and Scarpa obtained
quadratic gap dependence \cite{CSFree}; Chung, Wu, and Yuen improved the
gap power to $3/2$ for sufficiently many repetitions and extended the
result to multiple players \cite[Theorem~1.1]{CWY}.
Chailloux and Scarpa also proved exponential decay for correlated
question laws with full support \cite{CSFullSupport}.
Yuen obtained polynomial decay
for general entangled games \cite{Yuen}. OpenAI subsequently proved
exponential decay for arbitrary question distributions with gap power
thirteen \cite[Chapter~6]{OpenAI}, and Song improved the general gap
power to three \cite{Song}. OpenAI also established threshold repetition
for arbitrary finite entangled games \cite{OpenAIThreshold}.

We give a different proof of the cubic bound for arbitrary question
distributions. The main contribution is a choice of conditional comparison
states (Definition~\ref{def:conditional-comparisons}) and a conversion of
their quantum conditional relative entropies into compatible pure-state
descriptions, one for each player
(Lemma~\ref{lem:smooth-preparation}\ref{part:preparation-descriptions}).
The comparisons preserve the original conditional question law. Our chain
rule expresses each coordinate charge as the difference of two successive
relative entropies, so summing the charges cancels the intermediate terms
and leaves a single parent bound (Lemma~\ref{lem:ordered-reference-loss}).
The descriptions then enable local preparation of the conditioned workspace
(Lemma~\ref{lem:smooth-preparation}\ref{part:preparation-sampling}).

\subsection{All-win and threshold repetition}
Throughout, $\log$ denotes the base-two logarithm and $\ln$ the natural logarithm.
Let $G$ be a finite two-player game with question law $\mu$, answer
alphabets $\cA,\cB$, and entangled value $\omega^*(G)$
(Definition~\ref{def:game}). Write
\[
 \varepsilon:=1-\omega^*(G),\qquad
 L:=\max\{1,\log(|\cA||\cB|)\}.
\]
The entangled supremum need not be attained. We prove the following bound.
\begin{theorem}[Cubic repetition using Song's sampler]
\label{thm:smooth-repetition}
There is a universal constant $c>0$ such that, for every finite
 two-player game with $\varepsilon>0$ and every $n\ge1$,
\[
 \omega^*(G^{\otimes n})\le\exp[-c\varepsilon^3n/L].
\]
\end{theorem}
The same analysis also bounds the probability of winning a large
fraction of the coordinates.

\paragraph{Threshold repetition.}
Let $N$ be the number of coordinates won. For every $0<\delta\le\varepsilon$,
Corollary~\ref{cor:threshold-repetition} gives a universal $c_{\rm th}>0$ such that
\[
 \Pr\!\left[N\ge\left\lceil(1-\varepsilon+\delta)n\right\rceil\right]
 \le\exp\!\left[-c_{\rm th}\frac{\delta^5n}{\varepsilon^2L}\right]
\]
for every finite-dimensional repeated strategy. The choice $\delta=\varepsilon$
gives the all-win event. Taking $\delta=\varepsilon/2$
allows an $\varepsilon/2$ fraction of failures while retaining cubic dependence
on $\varepsilon$. OpenAI proves a universal threshold rate
$\Omega(\delta^5/L)$ and a complementary $\Omega_\mu(\delta^3/L)$
rate whose constant depends on the question distribution
\cite[Theorem~1.1 and Corollary~5.2]{OpenAIThreshold}.
Our bound improves the universal rate's gap dependence by a factor
$1/\varepsilon^2$.

\subsection{Proof strategy: extraction and conditioning}
\label{sec:intro-local-construction}
\paragraph{The extraction objective.}
Our main reduction turns high average success on unchecked coordinates,
conditioned on checked wins, into high single-game success. The extracted
strategy receives fresh questions from $\mu$ and approximately prepares
the corresponding conditional workspace of a uniformly chosen unchecked
coordinate. Every legal single-game strategy rejects with probability
at least $\varepsilon$. Thus, sufficiently small preparation and
question-law errors force mean unchecked rejection probability
$\Omega(\varepsilon)$. We will control these errors by the cost of
conditioning on the checked wins and retaining the checked answers.

\paragraph{The two strategy states.}
Fix a finite-dimensional repeated strategy and a checked set
$C\subsetneq[n]$. Represent the strategy coherently, using a purified
shared resource, question amplitudes
$\sqrt{\mu^{\otimes n}(\mathbf x,\mathbf y)}$, and local measurement
isometries. Retain all outputs and coherent basis copies of the questions
and checked answers. Call the resulting pure state $\ket{\psi}$ the
\emph{preselection state}. Let $W_C$ be the event that all checked
coordinates are won, and let $\Pi_C$ project onto its winning transcripts.
The \emph{postselection state} is
\[
 \ket{\phi}:=\frac{\Pi_C\ket{\psi}}{\sqrt{p_C}},
 \qquad p_C:=\Pr[W_C]=\bra{\psi}\Pi_C\ket{\psi}>0.
\]
Throughout, the bare symbol associated with a unit ket denotes its
density operator. The measured question-and-answer laws of $\psi,\phi$
are $\mathsf P$ and
$\mathsf Q=\mathsf P(\,\cdot\mid W_C)$. Targets and comparisons will be
conditional marginals of these two states; Definition~\ref{def:two-global-states}
gives their precise registers.

We use relative entropy $D$ and mutual information $I$ in bits, and write
$P$ for purified distance; infidelity means $P^2$. Superscripts denote
reduced systems. Definition~\ref{def:information-quantities} gives these
conventions, including conditional mutual information. For classical
systems they reduce to the corresponding classical quantities.

\paragraph{The Raz--Holenstein argument for classical games.}
The classical proof reconstructs the conditioned view of one unchecked
coordinate using shared randomness \cite{Raz,Holenstein}. Write $X,Y$
for its questions and $A,B$ for the classical views the players must
reconstruct.
The relevant information control can be expressed by the conditional
mutual information $I(X:B\mid Y)$: it measures, in average relative
entropy, how much Bob's view depends on Alice's question once his own
question is known. Small charge makes this view close to one Bob can
sample from $Y$ alone. Alice uses the symmetric charge $I(Y:A\mid X)$.
Conditioning and chain rules bound these dependencies on average by the
cost of the checked wins and answers. A record containing the checked
transcript and one publicly chosen question in each other coordinate
breaks the dependence between the residual classical views. Correlated
sampling lets the players agree on this record, after which they
generate the remaining parts locally. Applying the repeated strategy
to these reconstructed views extracts a single-game strategy; iterating
the resulting rejection bound yields exponential decay
\cite[Section~3, Corollary~4.3, and Lemmas~5.2 and~5.4]{Holenstein}.

\paragraph{JPY's argument for product questions.}
JPY use this conditioning approach for product question distributions,
following Raz and Holenstein \cite[Section~I, ``Our techniques'']{JPY}.
For one unchecked coordinate, let $h$ record the checked questions and
answers and the earlier question pairs. Write $\phi_h$ for the resulting
conditional state. Suppress the coordinate index, writing $X,Y$ for its
questions, $\hX,\hY$ for their copies, and $A,B$ for the remaining
workspaces. JPY's coherent charges have the form
\[
 I(\hY:\hX XA)_{\phi_h}
 \qquad\text{and}\qquad
 I(\hX:\hY YB)_{\phi_h},
\]
averaged over the history. They measure mutual information between each
player's coherent local view and the other player's question.
With the checked questions fixed and the checked-answer
registers traced out, the preselection marginal of $\psi$ is a mixture
over checked answers. Each winning branch is also the corresponding
conditional marginal of $\phi$, and its relative entropy to the mixture
is at most minus the logarithm of its mixing weight. Averaging over
winning branches bounds the parent divergence by
$B_C:=\log(1/p_C)+|C|L$: the first term is the conditioning cost,
and the second bounds the entropy of the checked answers
\cite[Lemma~III.1]{JPY}. Fixing the checked questions in both states
incurs no cost depending on their alphabet.

Product structure and a chain rule allocate this bound to the
unchecked-coordinate charges. When the budget per unchecked coordinate
is small, Uhlmann's theorem (Fact~\ref{fact:local-transitions}) supplies
input-dependent local unitaries;
the product question law allows the two corrections to be combined.
After accounting for the embedding error and the changed question law,
conditional success too close to one would give a single-game strategy
with success above $\omega^*(G)$. This forces an unchecked coordinate
with a definite rejection probability. Iterating the conditioning step
yields exponential decay.

\subsection{The conditional relative-entropy charge}
\label{sec:intro-information-model}
To extend this argument to correlated questions, we need a charge that
both controls local preparation and admits a budget bounded by conditioning.
We test candidate charges on general families of conditional workspace
states, asking whether the players can prepare them from an
input-independent shared state using only their own questions.

\paragraph{Why the Raz--Holenstein charge does not suffice.}
For quantum workspaces, the classical charge can miss phase dependence.
For independent uniform question bits $x,y$, take the phase version of
Chailloux and Scarpa's CHSH example \cite[Section~3.1]{CSFullSupport}:
\[
 \ket{\phi_{xy}}^{AB}
 :=\frac{\ket{00}+(-1)^{xy}\ket{11}}{\sqrt2}.
\]
The corresponding state with coherent question copies is
\[
 \ket{\phi}=\frac12\sum_{x,y\in\{0,1\}}
 \ket{xx}^{\hX X}\ket{yy}^{\hY Y}\ket{\phi_{xy}}^{AB}.
\]
The relative sign changes only when both questions are $1$. Yet Alice's
and Bob's individual workspace states are always $I/2$. Hence
$I(X:B\mid Y)_\phi=I(Y:A\mid X)_\phi=0$.
Keeping the question copies coherent does detect it: JPY's charge
distinguishes this family from the constant Bell-state family
$(\ket{00}+\ket{11})/\sqrt2$, even though their individual workspace
states are identical.\footnote{The two coherent mutual informations in
JPY's charge equal $1/2$ bit for the phase family and vanish for the
constant family. For Alice, each $y$-conditional state on $\hX XA$ has
entropy one, while their average has nonzero eigenvalues $1/2,1/4,1/4$
and hence entropy $3/2$. Their difference is $3/2-1=1/2$ bit. Bob's
calculation is the same with the players exchanged.}

\paragraph{Why the JPY charge does not suffice.}
For correlated questions, JPY's mutual information already charges
correlations in the original question law, even before conditioning.
For example, consider an unconditioned coordinate with
$\mu(0,0)=\mu(1,1)=1/2$ and fixed workspace $\ket{00}^{AB}$.
Then $I(\hY:\hX XA)_\phi=1$, although preparing $AB$ requires no
input-dependent operation. Over independent coordinates this charge is
linear even for the untouched source. The comparison must therefore
retain the prior correlation, rather than demand independence from the
remote question.

\paragraph{Why some natural alternatives do not suffice.}
One might try conditioning a coherent mutual information on the local
question to remove the prior correlations. Two natural choices fail for
different reasons. First, consider $I(\hX X:B\mid Y)$ and its Alice
counterpart. Enlarge the phase example by adding registers $A_{\rm c}$
and $B_{\rm c}$ to the players' workspaces, containing copies of their
own questions. With $A':=AA_{\rm c}$ and $B':=BB_{\rm c}$, the full state is
\[
 \ket{\widetilde\phi}
 :=\frac12\sum_{x,y\in\{0,1\}}
 \ket{xx}^{\hX X}\ket{yy}^{\hY Y}
 \ket{\phi_{xy}}^{AB}\ket{x}^{A_{\rm c}}\ket{y}^{B_{\rm c}}.
\]
In the marginal $\widetilde\phi^{\hX X Y B'}$, tracing out $A_{\rm c}$
makes $\hX X$ classical. Conditioned on $Y=y$, Bob's state is
$(I_B/2)\otimes\proj y^{B_{\rm c}}$, independent of $x$. Therefore
\[
 I(\hX X:B'\mid Y)_{\widetilde\phi}
 =I(\hY Y:A'\mid X)_{\widetilde\phi}=0.
\]
JPY's charges remain $1/2$ bit: adding these copies is a local isometry
of each full coherent view. Yet preparing the enlarged family locally
with arbitrarily small error would allow success approaching one in
CHSH. Discard the added copies and measure the original qubits $A,B$ in
the Hadamard basis; the ideal outputs satisfy $a\oplus b=xy$.
Tsirelson's bound limits entangled success to $\cos^2(\pi/8)<1$
\cite[p.~94]{Tsirelson1980}. Thus these conditional mutual informations,
like the classical charges, can vanish despite an obstruction to local
preparation.

A second choice, $I(\hY\hX:A\mid X)$ and its Bob counterpart, has the
opposite problem: it charges coherence created by copying one's own
question, even for product questions. For independent uniform bits, take
$\ket{\phi_{xy}}^{AB}=\ket{x}^A\ket{y}^B$, which the players prepare
locally. In its coherent copied-input state,
\[
 I(\hY\hX:A\mid X)_\phi=I(\hX\hY:B\mid Y)_\phi=1.
\]
Indeed, the state on $\hX XA$ is
$(\ket{000}+\ket{111})/\sqrt2$, and $\hY$ is independent of it.
Its two-register marginals and its $X$ marginal have entropy one, giving
$I(\hY\hX:A\mid X)_\phi=1+1-1-0=1$; Bob is symmetric.
Each player's total charge is therefore $n$ on $n$ untouched coordinates,
although $B_\varnothing=0$. We need a charge that detects the phase
obstruction while allowing the original question correlations and local
copying at no cost.

\paragraph{Our charge: quantum conditional relative entropy.}
Our replacement compares the coherent local view with a state that
preserves the original conditional question law and allows local
dependence on the player's own question. Subtracting the divergence of
the retained local marginal yields the quantum conditional relative
entropy used in both the preparation estimate and the conditioning budget.

In the one-coordinate view of $\phi$, let $q$ be the actual question
law, supported on the prescribed law $\mu$, and let
$\ket{\phi_{xy}}^{AB}$ be the conditional workspace vectors.
Alice holds $X,A$ and Bob holds $Y,B$. This gives the copied-input form
of Definition~\ref{def:copied-input-family}:
\[
 \ket{\phi}=\sum_{xy}\sqrt{q(x,y)}
 \ket{xx}^{\hX X}\ket{yy}^{\hY Y}\ket{\phi_{xy}}^{AB},
 \qquad R=\hX\hY.
\]
The players leave the reference $R$ untouched. Here
$\phi_{xy}^{AB}=\phi_{X=x,Y=y}^{AB}$, with the fixed question copies
suppressed. The charge and preparation estimates below apply to any
family of this form; in the repetition argument they are applied to
coordinate views of the postselection state.

An Alice conditional comparison is a state $\nu_A^{RXA}$ admitting a
purification of the form
\[
 \ket{\xi_A}
 =\sum_{x,y}\sqrt{r_X(x)\mu(y\mid x)}
 \ket{xx}^{\hX X}\ket{yy}^{\hY Y}\ket{\xi_{A,x}}^{AF},
 \qquad \nu_A=\Tr_{YF}\proj{\xi_A},
\]
for a distribution $r_X$ and unit vectors $\ket{\xi_{A,x}}$
(Lemma~\ref{lem:conditional-comparison-forms}). The conditional vector
uses only $x$; the prescribed kernel supplies its dependence on $y$.
Bob's comparisons interchange the players. The local marginal remains
free because conditioning can change it, and preparation later matches
it to the target's marginal. We prescribe the original conditional
kernels because the repetition construction supplies comparisons of this
form and preserves it under ordered restrictions. The comparison retains the
coherent $x,x'$ blocks needed for purification. On diagonal inputs it is the usual conditional form
$\nu_A(x,y,a)=\mu(y\mid x)\nu_A(x,a)$.

For a comparison pair with finite full divergences, define
(Definition~\ref{def:local-losses})
\[
 \begin{aligned}
 b_A&=D(\phi^{RXA}\|\nu_A)-D(\phi^{XA}\|\nu_A^{XA}),\\
 b_B&=D(\phi^{RYB}\|\nu_B)-D(\phi^{YB}\|\nu_B^{YB}),
 \qquad s:=b_A+b_B.
 \end{aligned}
\]
These are quantum conditional relative entropies, in the sense of Capel, Lucia, and P\'erez-Garc\'ia
\cite[Theorem~2]{CLPConditional}, for the reference $R$ conditioned on
the local question and workspace.
Data processing makes them nonnegative (Fact~\ref{fact:master-standard}). A fixed workspace
with $q=\mu$ is reproduced exactly by taking $r_X=\mu_X$ and the
same conditional vector at every $x$ in the defining purification.
Both $b_A$ and $b_B$ therefore vanish even for a correlated source, including
the correlated-bit example above. For the locally prepared family
$\ket{x}^A\ket{y}^B$, keep $r_X=\mu_X$ and take
$\ket{\xi_{A,x}}=\ket{x}^A\ket{0}^F$ and use the symmetric choice for
Bob; both comparisons again match their targets and both charges vanish.

The charges recover familiar quantities in two special cases. For product
$\mu=\mu_X\otimes\mu_Y$, every Alice comparison is independent of
$\hY$. Optimizing over these comparisons gives
\[
 \inf_{\nu_A} b_A
 =I(\hY:\hX XA)_\phi+D(q_Y\|\mu_Y).
\]
Lemma~\ref{lem:benchmark-loss-identities} proves this identity by the
product expansion of relative entropy. Thus $b_A$
recovers JPY's coherent charge and the change in the remote question
marginal. In particular, the minimum values of $b_A,b_B$ are $1/2$
for the phase example and zero for the constant Bell-state family.
The identity allows a correlated posterior $q$; the allocation itself
keeps the inherited comparisons fixed.

For a diagonal target, write $q^{XYA}$ for its classical law. With any
fixed diagonal comparison, the equality case of
Lemma~\ref{lem:loss-cq-information} gives
\[
 b_A=I(Y:A\mid X)_q+D(q^{XY}\|q_X\mu_{Y|X}).
\]
This is the classical conditional-information charge underlying the
Raz--Holenstein approach \cite{Raz,Holenstein}, together with the change
in the conditional question law. Lemma~\ref{lem:loss-cq-information}
gives the corresponding lower bound for classical questions and a
quantum workspace.

\subsection{Conditioning and allocation of the charge}
\label{sec:intro-chain-rule-budget}
We now show that the coordinate charges are small on average.
Put $J=[n]\setminus C$ and $m=|J|$, and let $q_i$ be the $(X_i,Y_i)$
marginal of $\mathsf Q$ (Definition~\ref{def:conditioned-strategy}).
Each coordinate view of $\phi$ has the form above with $q=q_i$.

The parent bound in the JPY overview also holds for correlated questions.
Compare a marginal of $\phi$ conditioned on the checked questions and
answers with the corresponding marginal of $\psi$ conditioned only on
the checked questions. Both marginals omit the checked-answer registers;
the target conditions on their values, while the comparison averages
over them.
Lemma~\ref{lem:conditional-marginal-budget} bounds the mean divergence
by $B_C$, also after taking either player's marginal. Put $t:=B_C/m$.

Order the unchecked coordinates increasingly. Alice's chain rule uses
successive prefix marginals of the question references; Bob's uses
successive suffix marginals. At coordinate $i$, these complementary
choices expose the same question history: Bob's questions before $i$
and Alice's questions after $i$, together with the checked questions.
The full history also records the checked answers. Both recorded states
retain this same history register, with the same postselection weights.
Within each history block, the target conditions on the checked answers
and the comparison averages over them. Both conditional blocks omit
the checked-answer registers and act on the same residual systems
(Definition~\ref{def:inherited-comparisons}). Only the target block
generally depends on the checked-answer values. The comparison blocks
are inherited marginals of $\psi$, fixed throughout the allocation.
Independence across the repeated question pairs ensures that the
restrictions preserve the original conditional question law.

Each coordinate charge, averaged over its history under $\mathsf Q$,
is the difference between two successive relative entropies. Summing
cancels the intermediate terms and leaves the parent divergence minus
a nonnegative marginal divergence. Thus each player's total charge is
at most $B_C$ (Lemma~\ref{lem:ordered-reference-loss}). Writing
$b_{A,i},b_{B,i}$ for the resulting charges,
Proposition~\ref{prop:game-information-interface} gives
\begin{align*}
 \frac1m\sum_{i\in J}(b_{A,i}+b_{B,i})
 &\le2t &&\text{(the two telescoping budgets)},\\
 \frac1m\sum_{i\in J}D(q_i\|\mu)
 &\le\frac{\log(1/p_C)}m\le t
 &&\text{(postselection and product chain rule).}
\end{align*}
The conditional charges control preparation under $q_i$. The separate
$D(q_i\|\mu)$ bound pays for the change to fresh questions from $\mu$
during extraction. It is needed because the conditional charges need
not control the joint question law for correlated $\mu$.
Section~\ref{sec:main-parent-allocation} gives the register assignments and
verifies the comparison form and the allocation identities.

\subsection{From the information budget to local preparation}
The local preparation error $d_{\rm prep}(q,\phi)$ is the infimum of the
mean infidelity with the family $\phi_{xy}^{AB}$ achievable from an
input-independent shared state by Alice's $x$-dependent and Bob's
$y$-dependent channels (Definition~\ref{def:local-preparation-error}).
The designated output includes the entire original game workspace;
additional implementation environments may be discarded.

Let
$\bar e$ denote preparation infidelity averaged uniformly over $i\in J$
and then under the question law $q_i$. The local preparation conclusion of
Lemma~\ref{lem:smooth-preparation}\ref{part:preparation-sampling} is
\begin{equation}
 d_{\rm prep}(q,\phi)
 \le\min\{1,C_{\rm prep}\sqrt{b_A+b_B}\},
 \label{eq:intro-loss-to-preparation}
\end{equation}
where $C_{\rm prep}$ is universal. The lemma specifies the constant and
gives a finite witness for every positive slack.
For every preparation slack $\epsilon_{\rm prep}>0$, Jensen's inequality
and the allocated budget provide
coordinate preparations with $\bar e\le C_{\rm prep}\sqrt{2t}+\epsilon_{\rm prep}$.
It remains to construct the local descriptions that give
\eqref{eq:intro-loss-to-preparation}.

\paragraph{The technical contribution: two local pure descriptions.}
Each player forms a pure-state description using only their own question.
Small charges make the two descriptions close on average, and local
corrections bring Alice's description close to the desired workspace.
Lemma~\ref{lem:smooth-preparation}\ref{part:preparation-descriptions}
makes these two requirements quantitative. For any finite copied-input
family with conditional comparisons and finite full divergences, it
supplies unit vectors $\ket{u_x},\ket{v_y}$ on a common space
$A_{\rm s}F$, Alice isometries $J_x:A_{\rm s}\to A$, and Bob unitaries
$W_y:B\to F$, after finite pure padding, such that
\[
 \boxed{\begin{aligned}
 \mathbb E_q\|\ket{u_x}-\ket{v_y}\|_2^2
   &\le4(\ln2)(b_A+b_B),\\
 \mathbb E_qP^2\!\left((J_x\otimes W_y^\dagger)u_x(J_x^\dagger\otimes W_y),\phi_{xy}\right)
   &\le(\ln2)b_A.
 \end{aligned}}
\]
All descriptions and corrections are fixed from the finite family
before the questions arrive.

The first estimate permits Song's sampler
(Fact~\ref{fact:song-smooth-sampling}) to prepare $u_x$ approximately.
Alice and Bob then apply $J_x$ and $W_y^\dagger$; the second estimate
connects that corrected state to the original workspace target.
The purified-distance triangle inequality then gives
\eqref{eq:intro-loss-to-preparation}.

\paragraph{Why the pure descriptions exist.}
The proof of Lemma~\ref{lem:smooth-preparation}\ref{part:preparation-descriptions}
first uses constructive marginal replacement
(Lemma~\ref{lem:retained-marginal-lift}\ref{part:replacement-conditional}).
It matches Alice's retained marginal while preserving the conditional
comparison form and controlling the error by $b_A$. Uhlmann's theorem,
applied to the whole coherent $x$ ensemble at fixed $y$, chooses one Bob
correction for all its components. Cauchy--Schwarz turns the coherent
overlap into a conditional-infidelity bound under $q$. Bob's information
charge then bounds the discrepancy between the canonical pure
descriptions.

\subsection{Completing extraction and repetition}
\label{sec:intro-complete-extraction}
We can now complete the extraction described at the start of the overview.
The targets of Definition~\ref{def:coordinate-losses} have designated
answer registers that reproduce the corresponding coordinates of the
conditioned strategy. In Procedure~\ref{proc:cubic-single-copy},
the players share a uniform unchecked coordinate and a preparation
resource for that coordinate, independent of the fresh questions.
On fresh questions they apply their respective local maps and read the
designated answers. The history is part of the prepared workspace, not
additional advice. Every such legal strategy
rejects with probability at least $\varepsilon$.

Let $A_C$ be the mean unchecked rejection probability under $\mathsf Q$.
The square-root rejection comparison
(Lemma~\ref{lem:square-root-rejection}) gives
\[
 \sqrt\varepsilon\le
 \sqrt{A_C}+\sqrt{\bar e}
 +\sqrt{\frac{\ln2}{m}\sum_{i\in J}D(q_i\|\mu)}.
\]
The last term accounts for the change from $q_i$ to $\mu$ by
purified-distance data processing under the fixed rejection channel.
Substituting the preparation and question-law bounds and letting
$\epsilon_{\rm prep}\downarrow0$ forces $A_C=\Omega(\varepsilon)$ whenever $t$ is at most a sufficiently
small universal multiple of $\varepsilon^2$.
Greedy checking selects a coordinate with rejection probability
$\Omega(\varepsilon)$ while this budget condition holds. Each check
reduces $p_C$ by a factor $1-\Omega(\varepsilon)$.
The proof of Theorem~\ref{thm:smooth-repetition} balances these decreases
against the growth of $B_C$ to obtain
$\exp[-\Omega(\varepsilon^3n/L)]$. The quadratic information threshold
and the linear rejection threshold account for the three powers of the
gap. Section~\ref{sec:game} makes this reduction quantitative.

The same extraction estimate gives threshold repetition. Keeping the
square-root rejection comparison forces $A_C\ge\varepsilon-\delta/4$
when $t$ is at most a sufficiently small universal multiple of
$\delta^4/\varepsilon^2$. Rao's concentration argument
\cite{RaoConcentration} turns this conditional-success bound into the
stated threshold rate. Corollary~\ref{cor:threshold-repetition} gives
the parameter choices.

\subsection{Comparison with the previous proofs}
\label{sec:intro-openai}
Table~\ref{tab:proof-comparison} compares our reconstruction with
OpenAI's alignment framework
\cite[Chapter~6, Sections~3.2--3.4 and~4.1--4.3]{OpenAI},
as refined by Song~\cite{Song}. In that approach, a reveal process
controls changes in local descriptions through operator-entropy
increments and a one-step variation bound for resolvent purifications.
Classical histories are coupled before a state depending on the history
is sampled. Song improves the sampling step to obtain the cubic rate.
Our proof uses the same sampler, but obtains its pure descriptions from
conditional comparisons: their relative entropies telescope to the
conditioning budget, and marginal replacement followed by Uhlmann's
theorem supplies the local corrections. We include the history registers
in the state to be prepared, so the local preparation step also prepares
their correlations with the remaining workspace.

\begin{proofcomparisontable}
\begin{tabularx}{\textwidth}{@{}>{\raggedright\arraybackslash}p{.16\textwidth}
 >{\raggedright\arraybackslash}X >{\raggedright\arraybackslash}X@{}}
\toprule
Task & Alignment route retained by Song & Present proof \\
\midrule
Information allocation & Reveal martingales and operator-entropy increments.
 & Forward and reverse chain rules for conditional relative entropies to fixed comparisons. \\
Local descriptions & Resolvent purifications and a one-step variation bound.
 & Constructive marginal replacement and a coherent local transition. \\
History and output & Couple classical histories, then sample a history-dependent state.
 & Include the histories in the target and prepare its full conditional workspace. \\
\bottomrule
\end{tabularx}
\caption{Construction of the local descriptions and treatment of history.
Both arguments use Song's sampling bound
(Fact~\ref{fact:song-smooth-sampling}).}
\label{tab:proof-comparison}
\end{proofcomparisontable}

\paragraph{Earlier uses of conditional relative entropy.}
Classically, the charge is an expected conditional relative entropy,
used by Jain in strong direct-product results \cite{JainNewSDP}; the
robust conditional relative min-entropy in that work is its distinct
tail analogue. Replacing the partial trace by a general channel gives
the data-processing remainder studied by Sutter, Tomamichel, and Harrow;
Junge, Renner, Sutter, Wilde, and Winter; and Sutter, Berta, and Tomamichel
\cite{STH,JRSWW,SBT}. Anshu, Jain, Mukhopadhyay, Shayeghi, and Yao use
the same channel loss in compression with classically described side information
\cite[Theorem~3.12]{AJMSY}.

\paragraph{Further related work.}
\label{sec:intro-classical-jpy}
Song's sampler refines the robust quantum correlated sampling developed
by Dinur, Steurer, and Vidick for entangled parallel repetition
\cite{DSV,Song}. The quantum
correlated-sampling problem of Anshu, Jain, Mukhopadhyay, Shayeghi, and
Yao has a different output requirement \cite[Theorem~5.1]{AJMSY}. Bavarian, Vidick,
and Yuen prove repetition via
anchoring and fortification, which modify the game
\cite{BVYAnchor,BVYFort}; the games repeated here are unchanged.
Kundu and Lalonde proved repetition bounds for shared EPR pairs subject
to depolarizing, unital, or erasure noise \cite[Theorem~4]{KLNoisy}.
Schmidt, Storgaard, Walter, and Zhao determined the exact entangled
values of all even parallel repetitions of Feige's game
\cite[Section~6]{SSWZFeige}.

\Needspace{8\baselineskip}
\section{Preliminaries and information-theoretic tools}
\label{sec:preliminaries}
All systems are finite dimensional. We fix the notation and collect the
general estimates used for conditioning, local preparation, and extraction.
Application-specific states and comparisons are constructed in
Section~\ref{sec:main-parent-allocation}.

\subsection{Games and information conventions}
\begin{definition}[Games, values, and gap parameters]
\label{def:game}
\normalfont
A finite two-player game is $G=(\cX,\cY,\cA,\cB,\mu,v)$, with
question law $\mu$ on $\cX\times\cY$, answer alphabets $\cA,\cB$,
and predicate $v(x,y,a,b)\in\{0,1\}$. Its entangled value
$\omega^*(G)$ is the supremum success probability over finite-dimensional
shared states and local measurements. Its classical value $\omega(G)$
allows shared and local randomness. The repeated game $G^{\otimes n}$
draws $n$ independent question pairs and accepts if every coordinate is won.
Write
\[
 \varepsilon:=1-\omega^*(G),\qquad
 L:=\max\{1,\log(|\cA||\cB|)\}.
\]
\end{definition}

\begin{definition}[Entropy and distance conventions]
\label{def:information-quantities}
\normalfont
For a density operator $\rho$ and a nonzero positive operator $\nu$ on the same space, define
\[
 D(\rho\|\nu):=\Tr\rho(\log\rho-\log\nu),
\]
where $D(\rho\|\nu)=+\infty$ unless $\supp\rho\subseteq\supp\nu$.
Whenever we subtract a marginal divergence from a full-state divergence,
this support inclusion is assumed, so both terms are finite.
The measured relative entropy is
\[
 D_{\mathbb M}(\rho\|\nu)
 :=\sup_{\{M_j\}}\sum_j(\Tr\rho M_j)
       \log\frac{\Tr\rho M_j}{\Tr\nu M_j},
\]
where the supremum is over finite POVMs: $M_j\ge0$ and $\sum_jM_j=I$.
A summand is zero when its numerator is zero and is $+\infty$ when
only its denominator is zero. We also use the max-relative entropy
\cite[Definition~1]{Datta},
\[
 D_{\max}(\rho\|\nu)
 :=\inf\{\lambda\in\mathbb R:\rho\le2^\lambda\nu\},
\]
with value $+\infty$ if no such $\lambda$ exists.
For normalized states, root fidelity and purified distance are
\[
 F(\rho,\nu):=\|\sqrt\rho\sqrt\nu\|_1,\qquad
 P(\rho,\nu):=\sqrt{1-F(\rho,\nu)^2}.
\]
For a unit vector $\ket{\psi}$, the same lowercase symbol denotes its
density operator $\psi:=\proj\psi$. Marginals and conditioned marginals
keep that name, whether pure or mixed; ket notation is used only for
pure states. Both arguments of $F$ and $P$
are density operators; in particular, an isometry acts there as
$\rho\mapsto U\rho U^\dagger$.
Infidelity means $P^2$.
We use $D$ for information budgets and $P^2$ for preparation error.
The square-root rejection comparison also uses $P$ to change the
classical question law.
For a state $\zeta$, define
\[
 \begin{aligned}
 H(S)_\zeta&:=-\Tr\zeta^S\log\zeta^S,\\
 H(S\mid T)_\zeta&:=H(ST)_\zeta-H(T)_\zeta,\\
 I(S:T)_\zeta&:=H(S)_\zeta+H(T)_\zeta-H(ST)_\zeta,\\
 I(S:T\mid K)_\zeta&:=H(SK)_\zeta+H(TK)_\zeta-H(K)_\zeta-H(STK)_\zeta.
 \end{aligned}
\]
Superscripts denote marginals. We write $z\leftarrow p$ for sampling $z$
from the distribution $p$.
\end{definition}

\begin{definition}[Matching-copy support]
\label{def:matching-support}
\normalfont
A matching pair $\hX X$ is supported on
$\operatorname{span}\{\ket{xx}:x\in\cX\}$. For a state on $\hX XA$,
this means
\[
 \sigma=\Pi_X\sigma\Pi_X
 =\sum_{x,x'}\ket{xx}\bra{x'x'}\otimes\sigma_{xx'}^A,
 \qquad
 \Pi_X:=\sum_x\proj{x,x}^{\hX X}\otimes I_A.
\]
\end{definition}
Off-diagonal $x,x'$ blocks are allowed; matching support does not mean
that the state is classical. Tracing $\hX$ makes $X$ classical, and
$X$-controlled operators preserve matching support.
A hat denotes a basis-copy register; protected references are specified separately.

\begin{definition}[Conditional states]
\label{def:conditional-states}
\normalfont
For a state $\sigma^{ZQ}$ and a distinguished basis $\{\ket{z}\}$ of $Z$, put
\[
 p_\sigma(z):=\Tr\bra{z}\sigma\ket{z},
 \qquad
 \sigma_{Z=z}^{Q}:=\frac{\bra{z}\sigma\ket{z}}{p_\sigma(z)}
 \quad\text{when }p_\sigma(z)>0.
\]
Subscripts specify the records conditioned on; superscripts specify the
systems retained. By default, the conditioning register $Z$ is omitted.
If it is explicitly retained, then
$\sigma_{Z=z}^{ZQ}=\proj z^Z\otimes\sigma_{Z=z}^Q$.
Other registers containing copies of $z$ are retained or omitted
according to the indicated systems.
For a pure state, the corresponding conditional vector is
\[
 \ket{\sigma_{Z=z}}^Q
 :=(\bra{z}\otimes I_Q)\ket{\sigma}/\sqrt{p_\sigma(z)}.
\]
Once the conditioning registers are specified, write $\sigma_z^Q$ for
$\sigma_{Z=z}^Q$, and $\sigma_{xy}^Q$ for
$\sigma_{X=x,Y=y}^Q$. Subscripts explicitly introduced as family or
construction indices retain that meaning. In an expectation over $Z$,
$\sigma_Z^Q$ uses the sampled value but is normalized with $p_\sigma$,
even when the expectation uses another law. Fixed copies of conditioned
values may be suppressed. Omitting $Z$ means tracing it out, not
conditioning: $\sigma^Q=\sum_zp_\sigma(z)\sigma_{Z=z}^Q$.
\end{definition}

Unless stated otherwise, all outputs are retained, local partial
isometries are completed after finite padding, and zero-probability
branches receive fixed local completions. In limiting constructions,
approximation slack is chosen uniformly over the finite input family.

\Needspace{12\baselineskip}
\subsection{Copied-input families and conditional comparisons}
We isolate the preparation problem for one pair of questions. The
comparison states retain the prescribed conditional question law while
allowing coherent dependence on the local question.
\begin{definition}[Copied-input state family]
\label{def:copied-input-family}
\normalfont
Let $\cX,\cY$ be finite label sets and $\mu$ a prescribed distribution
on their product. Let $q$ be supported on $\supp\mu$. Alice holds the
control $X$ and workspace $A$; Bob holds $Y$ and $B$. All local records
belong to these workspaces. For bipartite unit vectors $\ket{\psi_{xy}}^{AB}$, define
\begin{equation}
 \ket{\psi}:=\sum_{xy}\sqrt{q(x,y)}
 \ket{xx}^{\hX X}\ket{yy}^{\hY Y}\ket{\psi_{xy}}^{AB},
 \qquad R:=\hX\hY.
 \label{eq:intro-native-state}
\end{equation}
Thus $\psi_{xy}^{AB}=\psi_{X=x,Y=y}^{AB}$ in the conditioning convention
of Definition~\ref{def:conditional-states}; the fixed question copies are
suppressed, and zero-probability vectors are arbitrary. The copied
questions $\hX,\hY$ form a protected reference on which the players never
act.
\end{definition}
When referring to a symmetric statement for Bob, we interchange
$(X,A,\hX,x)$ with $(Y,B,\hY,y)$ and use the corresponding conditional
question law. For ordered coordinate restrictions we also reverse the
coordinate order. Conclusions involving both players remain explicit.

\begin{definition}[$\mu$-conditional comparisons]
\label{def:conditional-comparisons}
\normalfont
For the alphabets and prior of Definition~\ref{def:copied-input-family}, put
\[
 K_y^\mu:=\sum_{x\in\supp\mu_X}\sqrt{\mu(y\mid x)}\proj x^X,
\]
with identities on the other registers. An \emph{Alice
$\mu$-conditional comparison} is a state of the form
\begin{equation}
 \nu_A^{RXA}=\mathcal C_\mu^A(\tau_A)
 :=\sum_y\proj y^{\hY}\otimes K_y^\mu\tau_AK_y^\mu,
 \qquad R=\hX\hY,
 \label{eq:conditional-comparison-density-form}
\end{equation}
where $\tau_A^{\hX XA}$ is a density operator with matching-copy
support and $X$-law supported on $\supp\mu_X$. Bob's comparisons
interchange the players and use $\mu(x\mid y)$.
We omit $\mu$ when the prior is fixed. Conditional laws are used only
at positive-mass marginal symbols; cells of $\mu$ and eigenvalues of
the comparison may vanish.
\end{definition}
The map $\mathcal C_\mu^A$ attaches the prescribed conditional law while
retaining the coherent $x,x'$ blocks. This defines a comparison state;
it is not a player operation on the protected reference. The following
equivalent purification form is used for local reconstruction.

\begin{lemma}[Equivalent forms of $\mu$-conditional comparisons]
\label{lem:conditional-comparison-forms}
A state $\nu_A^{RXA}$ is an Alice conditional comparison if and only if
there are a distribution $r_X$ supported on $\supp\mu_X$, a finite
system $F$, and unit vectors $\ket{\xi_{A,x}}^{AF}$ such that
\begin{equation}
 \begin{aligned}
 \ket{\xi_A}
 &:=\sum_{x,y}\sqrt{r_X(x)\mu(y\mid x)}
 \ket{xx}^{\hX X}\ket{yy}^{\hY Y}\ket{\xi_{A,x}}^{AF},\\
 \nu_A^{RXA}&=\Tr_{YF}\proj{\xi_A}.
 \end{aligned}
 \label{eq:phase-native-comparison}
\end{equation}
For every density-form representation,
\[
 \nu_A^{XA}=\tau_A^{XA}
 =\sum_xr_X(x)\proj x^X\otimes\xi_{A,x}^A.
\]
The symmetric characterization holds for Bob.
\end{lemma}
\begin{proof}
\emph{From the density form to the purification form.}
Choose a purification $\ket{\eta_A}^{\hX XAF}$ of $\tau_A$.
Matching-copy support lets us write
$\ket{\eta_A}=\sum_x\ket{xx}^{\hX X}\ket{v_x}^{AF}$.
Set $r_X(x):=\|\ket{v_x}\|^2$ and, when $r_X(x)>0$, set
$\ket{\xi_{A,x}}:=\ket{v_x}/\sqrt{r_X(x)}$; choose any unit vector on
zero-weight blocks. Then $r_X$ is the $X$-law of $\tau_A$, and
\[
 \ket{\eta_A}
 =\sum_x\sqrt{r_X(x)}\ket{xx}^{\hX X}\ket{\xi_{A,x}}^{AF}.
\]
Adjoin $\hY,Y$ coherently using the conditional amplitudes
$\sqrt{\mu(y\mid x)}$. This gives the vector in
\eqref{eq:phase-native-comparison}, equivalently
\[
 \ket{\xi_A}
 =\sum_y\ket{yy}^{\hY Y}\otimes(K_y^\mu\otimes I_{\hX AF})\ket{\eta_A}.
\]
The conditional probabilities sum to one, so this vector is normalized.
Tracing $Y$ removes terms with different $y$, while tracing $F$ gives
$\tau_A$. Therefore
\[
 \Tr_{YF}\proj{\xi_A}
 =\sum_y\proj y^{\hY}\otimes K_y^\mu\tau_AK_y^\mu
 =\nu_A.
\]

\emph{From the purification form to the density form.}
Given $r_X$ and the unit vectors in \eqref{eq:phase-native-comparison},
form $\ket{\eta_A}$ by the first display and put
$\tau_A:=\Tr_F\proj{\eta_A}$. Its normalization, matching-copy support,
and $X$-law $r_X$ meet Definition~\ref{def:conditional-comparisons}.
The same partial-trace calculation gives
$\nu_A=\mathcal C_\mu^A(\tau_A)$.

Finally, tracing $\hX$ removes the off-diagonal $x,x'$ blocks, so
\begin{align*}
 \nu_A^{XA}
 &=\sum_{x,y}r_X(x)\mu(y\mid x)\proj x^X\otimes\xi_{A,x}^A
 &&\text{(tracing the question copies)}\\
 &=\sum_xr_X(x)\proj x^X\otimes\xi_{A,x}^A
 =\tau_A^{XA}
 &&\text{($\sum_y\mu(y\mid x)=1$).}
\end{align*}
\end{proof}
In particular, prescribing $\nu_A^{XA}=\psi^{XA}$ fixes $r_X=q_X$.
The joint choice of the vectors $\ket{\xi_{A,x}}$ determines coherent
$x,x'$ blocks. Replacing them by separately chosen canonical
purifications of their $A$ marginals can change $\nu_A$.

\begin{definition}[Local conditional relative entropies]
\label{def:local-losses}
\label{def:main-comparison-loss}
\normalfont
For $\psi$ from Definition~\ref{def:copied-input-family}
and comparisons from Definition~\ref{def:conditional-comparisons},
assume both full divergences below are finite and define
\begin{equation}
 b_A:=D(\psi^{RXA}\|\nu_A)-D(\psi^{XA}\|\nu_A^{XA}),
 \qquad s:=b_A+b_B.
 \label{eq:intro-native-losses}
\end{equation}
Here $b_B$ is defined by replacing $(XA,\nu_A)$ with $(YB,\nu_B)$.
Each quantity is the quantum conditional relative entropy in $R$,
conditioned on the local question and workspace~\cite[Theorem~2]{CLPConditional}.
Equivalently, it is the relative entropy lost when $R$ is traced out.
We use the displayed difference whenever the full divergence is finite,
including when the comparison is singular.
\end{definition}

\begin{definition}[Local preparation error]
\label{def:local-preparation-error}
\normalfont
For the copied-input family of Definition~\ref{def:copied-input-family},
a shared state $\eta$ and local channels $\mathcal A_x,\mathcal B_y$
produce
\[
 \sigma_{xy}^{AB}:=(\mathcal A_x\otimes\mathcal B_y)(\eta).
\]
Define
\begin{equation}
 d_{\rm prep}(q,\psi)
 :=\inf_{\eta,\{\mathcal A_x\},\{\mathcal B_y\}}
       \sum_{xy}q(x,y)
       P^2(\sigma_{xy}^{AB},\psi_{xy}^{AB}).
 \label{eq:intro-operational-preparation}
\end{equation}
The infimum is over finite input-independent shared states $\eta$ and
local channels with designated outputs $A,B$. The resource and all maps
are fixed from the target family before the questions arrive; Alice's
channel uses only $x$, and Bob's only $y$. Their outputs contain the
entire original game workspace. Additional implementation environments
may be traced out. No resource-dimension or circuit-size bound is imposed.
\end{definition}

\subsection{Distances and rejection probabilities}
The following estimates transfer state-preparation error to rejection
probability, including when the question law changes.

\begin{fact}[Data processing and distance inequalities]
\label{fact:master-standard}
Relative entropy $D$ and purified distance $P$ contract under channels
and are invariant under a common isometry
\cite{Lindblad,TCR}, \cite[Chapter~3]{Watrous}. For normalized states,
the same properties hold for the measured relative entropy $\DM$
(Definition~\ref{def:information-quantities}), by pulling measurements back
through the channel. Also $D,\DM\ge0$ and $0\le F,P\le1$. The bound
$D(\rho\|\nu)\le D_{\max}(\rho\|\nu)$ holds for normalized states
\cite[Lemma~10]{Datta}.
The purified distance satisfies the triangle inequality,
\[
 P(\rho,\tau)\le P(\rho,\sigma)+P(\sigma,\tau).
\]
For bipartite unit vectors $\ket u,\ket v$, one also has
$P^2(u,v)\le\|\ket u-\ket v\|_2^2$.
\end{fact}

\begin{fact}[Isometric dilation {\cite[Section~2.2]{Watrous}}]
\label{fact:isometric-dilation}
Every finite-dimensional channel has an isometric dilation with an
environment that may be retained.
\end{fact}

\begin{fact}[Norm inequalities and averaging
{\cite[Section~1.1.1 and Theorem~1.14]{Watrous}}]
\label{fact:averaging-inequalities}
Cauchy--Schwarz, Minkowski in $L^2$, and Jensen's inequality hold for
weighted finite sums and probability integrals, whenever the displayed
quantities are finite. In particular,
\[
 \|v+w\|^2\le2\|v\|^2+2\|w\|^2.
\]
The triangle inequality for $P$ similarly gives
$P^2(\rho,\tau)\le2P^2(\rho,\sigma)+2P^2(\sigma,\tau)$.
\end{fact}

\begin{fact}[Divergence and fidelity inequalities
{\cite[Eq.~(66) and preceding paragraph]{SBT};
\cite[Theorem~5.38]{Watrous}}]
\label{fact:divergence-fidelity}
For normalized states,
\begin{align*}
 F^2&\ge2^{-\DM}
 &&\text{(measured divergence--fidelity bound)}\\
 &\ge2^{-D}
 &&\text{(data processing for measurements).}
\end{align*}
In particular, $P^2\le1-2^{-D}\le(\ln2)D$.
\end{fact}

\begin{lemma}[Square-root rejection and changes of question law,
adapted from {\cite[Lemma~4.5]{Song}}]
\label{lem:square-root-rejection}
\label{fact:square-root-rejection}
Let $\rho,\sigma$ be density operators and $0\le N\le I$ a rejecting
effect. Put $\alpha:=\Tr N\sigma$, $\beta:=\Tr N\rho$, and
$\delta:=P^2(\rho,\sigma)$. Then
\[
 |\sqrt\beta-\sqrt\alpha|\le\sqrt\delta.
\]
For a finite family indexed by $z$, let $p,q$ be laws on the labels and put
\[
 a_q:=\mathbb E_q\alpha_z,\qquad
 e_q:=\mathbb E_q\delta_z,\qquad
 r_q:=\mathbb E_q\beta_z,\qquad r_p:=\mathbb E_p\beta_z.
\]
The states, effects, and spaces may depend on $z$, but are fixed when
changing the law. Then
\begin{equation}
 \sqrt{r_q}\le\sqrt{a_q}+\sqrt{e_q},
 \label{eq:mean-square-root-rejection}
\end{equation}
and, regarding $p,q$ as diagonal states,
\begin{equation}
 \boxed{\sqrt{r_p}\le\sqrt{a_q}+\sqrt{e_q}+P(p,q).}
 \label{eq:rejection-change-law}
\end{equation}
If $D(q\|p)<\infty$, the last term is at most
$\sqrt{(\ln2)D(q\|p)}$.
\end{lemma}
\begin{proof}
The pointwise inequality implies \eqref{eq:mean-square-root-rejection}
by Minkowski. Apply that same
pointwise inequality to the Bernoulli outputs of the fixed channel
$\mathcal K:z\mapsto(\beta_z,1-\beta_z)$ to obtain
\[
 |\sqrt{r_p}-\sqrt{r_q}|
 \le P(\mathcal K(p),\mathcal K(q))\le P(p,q).
\]
The second inequality is data processing (Fact~\ref{fact:master-standard}).
This proves \eqref{eq:rejection-change-law}; the divergence consequence
is Fact~\ref{fact:divergence-fidelity}.

It remains to prove the pointwise inequality. Apply $\{N,I-N\}$ and
write $\alpha=\sin^2a$, $\beta=\sin^2b$, with $a,b\in[0,\pi/2]$.
The output purified distance is
$\sin|a-b|\ge|\sin a-\sin b|$. Data processing bounds it by
$P(\rho,\sigma)=\sqrt\delta$.
\end{proof}

\begin{fact}[Block fidelity {\cite[Section~3.2]{Watrous}}]
\label{fact:block-fidelity}
\label{fact:block-fidelity-hellinger}
For probability distributions $p,q$ and density operators
$\rho_j,\sigma_j$,
\begin{equation}
 F\!\left(\bigoplus_j p_j\rho_j,\bigoplus_j q_j\sigma_j\right)
 =\sum_j\sqrt{p_jq_j}\,F(\rho_j,\sigma_j).
 \label{eq:classical-block-fidelity}
\end{equation}
Zero-probability blocks may be completed arbitrarily.
\end{fact}

\subsection{Canonical purifications}
\label{sec:canonical-sampling-tools}
The sampler requires nearby vectors in one fixed coordinate system.
Canonical purifications turn local density operators into such vectors;
Uhlmann's theorem supplies the corrections back to the desired target.
\begin{definition}[Canonical purification]
\label{def:canonical-purification}
\normalfont
Fix a reference basis and use the vectorization convention
\[
 \vket{M}^{AB}:=\sum_j\ket j^A\otimes M\ket j^B.
\]
For a density operator $\gamma=\sum_i\lambda_i\ket{a_i}\bra{a_i}$
on $B$, its canonical purification is
\[
 \vket{\sqrt\gamma}^{AB}=\sum_i\sqrt{\lambda_i}\,
 \ket{\overline{a_i}}_A\ket{a_i}_B.
\]
Here $\ket{\overline{a_i}}$ is the entrywise complex conjugate of
$\ket{a_i}$ in the fixed reference basis.
This vector is independent of the eigenbasis in degenerate eigenspaces.
The factor order fixes the convention: the second marginal is $\gamma$
and the first is $\gamma^T$.
\end{definition}

\begin{fact}[Uhlmann's theorem and purification equivalence
{\cite{Uhlmann}; \cite[Theorem~2.12]{Watrous}}]
\label{fact:local-transitions}
\label{fact:known-purification-correction}
Let $\ket{r}^{AB}$ and $\ket{s}^{AC}$ purify
$\rho^A$ and $\sigma^A$. After finite padding of the purifying systems,
Uhlmann's theorem supplies an isometry $U:C\to B$ such that
\[
 \bra{r}(I_A\otimes U)\ket{s}=F(\rho^A,\sigma^A)\ge0.
\]
For classical reference blocks the isometries may be selected separately.
In particular, if $\ket{\psi}^{AB}$ and $\vket{\sqrt\rho}^{A_0B}$
purify the same marginal $\rho^B$, there is an Alice-side isometry
$J_\psi:A_0\to A$, after padding, with
\[
 (J_\psi\otimes I_B)\vket{\sqrt\rho}=\ket{\psi}.
\]
Only the correcting party must know the desired purification. For every
approximate preparation $\sigma^{A_0B}$, isometry invariance gives
\[
 P\bigl((J_\psi\otimes I_B)\sigma(J_\psi^\dagger\otimes I_B),\psi\bigr)
 =P\bigl(\sigma,\vket{\sqrt\rho}\vbra{\sqrt\rho}\bigr).
\]
The symmetric statement holds for a Bob-side correction, with the
canonical purification's factor order interchanged.
\end{fact}

\begin{fact}[Canonical vector discrepancy
{\cite{Raggio1982}; \cite[Eq.~(4.2)]{LiangFidelity}}]
\label{fact:canonical-sampling-discrepancy}
For density operators $\gamma,\beta$ in the common vectorization
convention of Definition~\ref{def:canonical-purification},
\begin{align}
 \|\vket{\sqrt\gamma}-\vket{\sqrt\beta}\|_2^2
  &=2(1-\Tr\sqrt\gamma\sqrt\beta)
  &&\text{(fixed vectorization)}\notag\\
  &\le2P^2(\gamma,\beta)
  &&\text{(affinity--fidelity inequality).}
 \label{eq:canonical-sampling-discrepancy}
\end{align}
Thus a mean squared purified-distance bound $e$ for local density
descriptions supplies a mean squared vector discrepancy at most $2e$.
\end{fact}

This follows from the affinity--fidelity inequality
$F(\gamma,\beta)^2\le\Tr\sqrt\gamma\sqrt\beta$ in the fixed
vectorization convention. No input-pair-dependent phase is selected.

\subsection{Classical and classical--quantum information}
The following identities will account separately for conditioning on
checked answers and for changing the question distribution.

\begin{fact}[Classical and classical--quantum information identities
{\cite[Section~2.5]{CoverThomas}; \cite[Lemma~4.2]{Holenstein}}]
\label{fact:classical-kl-identities}
For a finite classical register $Z$, one has
$0\le H(Z\mid H)\le\log|\mathcal Z|$.
For finite probability distributions with the displayed divergences finite,
\begin{align*}
 D(p^{UV}\|q^{UV})
 &=D(p^U\|q^U)+\E_{p_U}D(p^{V|U}\|q^{V|U}),\\
 D(p^{Z_1\cdots Z_m}\|\textstyle\prod_j q_j)
 &=D(p^{Z_1\cdots Z_m}\|\textstyle\prod_j p^{Z_j})
          +\sum_jD(p^{Z_j}\|q_j).
\end{align*}
In particular the second expression is at least the sum of its
marginal divergences.
For density operators $\rho_h,\sigma_h$, the direct-sum identity is
\cite[Section~5.2]{Watrous}
\[
 D\!\left(\bigoplus_h p_h\rho_h\,\middle\|\,
               \bigoplus_h q_h\sigma_h\right)
 =D(p\|q)+\sum_h p_hD(\rho_h\|\sigma_h).
\]
For a classical--quantum state and its conditional averages,
\[
 \zeta^{XYF}=\sum_{x,y}q(x,y)\proj{xy}\otimes\zeta_{xy}^F,
 \qquad \beta_y:=\sum_xq(x\mid y)\zeta_{xy},
\]
we have
\[
 I(X:F\mid Y)_\zeta
 =\sum_{x,y}q(x,y)D(\zeta_{xy}\|\beta_y).
\]
All these formulas remain valid after conditioning on a fixed
classical record and averaging with the actual target distribution.
\end{fact}

\begin{fact}[Postselection and product marginals
{\cite[Section~4, especially Lemma~4.2]{Holenstein}}]
\label{fact:postselection}
If \(Q=P(\,\cdot\mid W)\) and \(P(W)=p>0\), then for every random variable
\(Z\),
\[
 D(Q_Z\|P_Z)\le \log(1/p).
\]
If \(P_{Z_1\cdots Z_t}=\prod_jp_j\), then
\begin{align}
 \sum_{j=1}^tD(Q_{Z_j}\|p_j)
 &\le D(Q_{Z_1\cdots Z_t}\|\textstyle\prod_jp_j)
 &&\text{(product chain rule)}\notag\\
 &\le\log(1/p)
 &&\text{(postselection).}
 \label{eq:product-postselection}
\end{align}
\end{fact}

\begin{fact}[Rao's concentration bound, conditional-success form
{\cite[Section~8, Lemma~31 and the subsequent counting argument]{RaoConcentration}}]
\label{fact:rao-concentration}
Let $W_1,\ldots,W_n$ be events on one probability space, and put
\[
 W_C:=\bigcap_{i\in C}W_i,\qquad p_C:=\Pr[W_C],\qquad
 N:=\sum_{i=1}^n\mathbf1_{W_i}.
\]
Fix $L\ge1$, $0\le v<1$, $0<\delta\le1-v$, and $0< t_{\rm th}\le\delta/4$.
Suppose that every proper $C\subsetneq[n]$ with $p_C>0$ satisfies
\begin{equation}
 \frac{|C|L+\log(1/p_C)}{n-|C|}\le t_{\rm th}
 \quad\Longrightarrow\quad
 \frac1{n-|C|}\sum_{i\notin C}\Pr[W_i\mid W_C]
 \le v+\frac{\delta}{4}.
 \label{eq:rao-conditional-hypothesis}
\end{equation}
One may take $K_{\rm R}=2$ and $c_{\rm R}=1/12$ in
\begin{equation}
 \Pr\!\left[N\ge\left\lceil(v+\delta)n\right\rceil\right]
 \le K_{\rm R}\exp\!\left[-c_{\rm R}\frac{ t_{\rm th}\delta n}{L}\right].
 \label{eq:rao-concentration}
\end{equation}
No independence of the events is assumed.
\end{fact}
The parameter correspondence and rounding checks are given in
Appendix~\ref{sec:rao-parameters}.

\subsection{Correlated sampling}
\label{sec:dsv-song-sampling}
The sampler turns two nearby local descriptions into a shared
approximation of Alice's description.

\begin{fact}[Song's smooth sampling, optimized form
{\cite[Lemma~4.1 and Remark~4.2]{Song}}]
\label{fact:song-smooth-sampling}
Let $\ket{u_x},\ket{v_y}$ be finite families of bipartite unit vectors on common
spaces and let $p$ be a distribution on their labels. Put
\[
 \mathcal D:=\mathbb E_p\|\ket{u_x}-\ket{v_y}\|_2^2.
\]
For every $\epsilon_{\rm sm}>0$, finite input-independent shared resources and local
channels depending only on $\ket{u_x}$ and $\ket{v_y}$, respectively, produce a
designated state $\sigma_{xy}$ such that
\begin{equation}
 \mathbb E_pP^2(\sigma_{xy},u_x)
 \le\min\{1,4C_{\rm sm}\sqrt{\mathcal D}\}+\epsilon_{\rm sm},
 \qquad C_{\rm sm}\ge1,
 \label{eq:song-optimized-powers}
\end{equation}
where $C_{\rm sm}$ is universal. The resources, output spaces, and local
maps are fixed for the whole family before the labels arrive.
By isometric dilation (Fact~\ref{fact:isometric-dilation}), these channels
may equivalently be implemented by local isometries with environments.
The displayed guarantee concerns the designated output after tracing
those environments, not the complete dilated state.
\end{fact}
For completeness, Song's parameterized estimate is
$C_{\rm sm}(h^2+\mathcal D/h^2+\mathcal D)+\zeta$ for
$0<h<1/2$ and arbitrary $\zeta>0$.
If $0<\mathcal D<1/16$, choose $h=\mathcal D^{1/4}$; then
$2\sqrt{\mathcal D}+\mathcal D\le3\sqrt{\mathcal D}$.
For $\mathcal D\ge1/16$ the cap follows from $P^2\le1$.
At $\mathcal D=0$, choose $h$ and $\zeta$ small enough for the desired
positive slack. Every choice uses finite resources.

\subsection{Conditioning and the parent divergence budget}
The JPY conditioning estimate applies to an arbitrary classical--quantum
state, without product questions or pure conditional states.

\begin{lemma}[Conditional-marginal budget, adapted from {\cite[Lemma~III.1]{JPY}}]
\label{lem:conditional-marginal-budget}
Let
\[
 \chi^{HZQ}:=\sum_{h,z}p(h,z)\proj{h,z}\otimes\chi_{h,z}^{Q}
\]
be a density operator with classical $H,Z$. Let $W$ be an event
determined by $(H,Z)$, with $p_W:=\Pr_p(W)>0$, and put
$q:=p(\,\cdot\mid W)$. For $p_H(h)>0$, its marginal conditioned only on $H=h$ is
\[
 \chi_h^Q:=\sum_zp(z\mid h)\chi_{h,z}^Q.
\]
Then the supported conditional states have finite divergence and
\begin{equation}
 \mathbb E_{(h,z)\leftarrow q}D(\chi_{h,z}^Q\|\chi_h^Q)
 \le\log(1/p_W)+\log|\mathcal Z|.
 \label{eq:conditional-marginal-budget}
\end{equation}
The same bound holds after taking any corresponding reduced systems
in the two arguments.
\end{lemma}
\begin{proof}
For every $q$-supported pair $(h,z)$, we have $p(z\mid h)>0$ and
\[
 \chi_h=\sum_{z'}p(z'\mid h)\chi_{h,z'}\ge p(z\mid h)\chi_{h,z}.
\]
Consequently,
\begin{align*}
 D(\chi_{h,z}\|\chi_h)
 &\le D_{\max}(\chi_{h,z}\|\chi_h)
 &&\text{(Fact~\ref{fact:master-standard})}\\
 &\le\log\frac1{p(z\mid h)}<\infty
 &&\text{(the operator inequality above).}
\end{align*}
By the classical chain rule
and postselection identity (Facts~\ref{fact:classical-kl-identities}
and~\ref{fact:postselection}),
\begin{align*}
 \mathbb E_q\log\frac1{p(Z\mid H)}
 &=H_q(Z\mid H)+\mathbb E_{q_H}D(q_{Z\mid H}\|p_{Z\mid H})
 &&\text{(conditional cross entropy)}\\
 &=H_q(Z\mid H)+\log(1/p_W)-D(q_H\|p_H)
 &&\text{(chain rule and postselection)}\\
 &\le\log|\mathcal Z|+\log(1/p_W)
 &&\text{(entropy bound and $D\ge0$).}
\end{align*}
Data processing (Fact~\ref{fact:master-standard}) gives the reduced-state assertion.
\end{proof}
The target conditions on $H,Z$, but the comparison only conditions on
$H$. The estimate charges the finite record $Z$ and the event $W$;
it does not charge the alphabet of $H$.

\subsection{Chain rule and ordered conditional comparisons}
\label{sec:preliminary-comparison-tools}
The following lemma combines the relative-entropy chain rule with the
restrictions used to construct a single coordinate. All conditional
states use Definition~\ref{def:conditional-states}; target and comparison
blocks are normalized separately.

\begin{lemma}[Retained records and ordered allocation]
\label{lem:loss-chain-rule}
\label{lem:ordered-reference-loss}
All displayed full divergences are assumed finite. Write
$R_{\le i}:=R_1\cdots R_i$ and $R_{<i}:=R_1\cdots R_{i-1}$, with
empty strings denoting a trivial system; use $R_{\ge i},R_{>i}$ similarly.
\begin{enumerate}[label=(\roman*),ref=(\roman*),leftmargin=*,itemsep=5pt]
\item\label{part:loss-classical-blocks} \emph{Retained classical records.}
If $\rho^{SRK},\nu^{SRK}$ are classical on $S$, with laws $p_s,q_s$
and normalized blocks $\rho_s^{RK},\nu_s^{RK}$, then
\begin{equation}
 \begin{aligned}
 D(\rho^{SRK}\|\nu^{SRK})-D(\rho^{SK}\|\nu^{SK})
 =\sum_s p_s\!\left[D(\rho_s^{RK}\|\nu_s^{RK})
                        -D(\rho_s^K\|\nu_s^K)\right].
 \end{aligned}
 \label{eq:loss-classical-conditioning}
\end{equation}
The average uses the target law $p_s$, not the comparison law $q_s$.
\item\label{part:ordered-telescope}\label{part:ordered-comparisons}\label{part:loss-telescope}
\emph{Ordered allocation and conditional comparisons.}
Let $R_j:=\hX_j\hY_j$. Suppose $\rho,\nu$ are states on $R_1\cdots R_mK$,
classical on $\hY^m$, with $X^m\subset K$ and matching $\hX_jX_j$ support.
A register $T$ in the remaining factors of $K$ may also be classical in both states.
Trace $R_{>i}$ and condition on $H_i=(T,\hY_{<i},X_{>i})$,
retaining fixed record copies. Each conditional state uses its own
normalization, and $\mathbb E_{H_i}$ uses the record law of $\rho$. Then
\begin{align}
 &\sum_{i=1}^m\mathbb E_{H_i}\!\left[
 D(\rho_{H_i}^{R_{\le i}K}\|\nu_{H_i}^{R_{\le i}K})
 -D(\rho_{H_i}^{R_{<i}K}\|\nu_{H_i}^{R_{<i}K})\right]\notag\\
 &\qquad=D(\rho\|\nu)-D(\rho^K\|\nu^K).
 \label{eq:ordered-reference-telescope}
\end{align}
If, moreover, $K=X^mA$ and
$\nu=\mathcal C_{\mu^{\otimes m}}^A(\tau)$ is an Alice conditional
comparison, every supported restriction is an Alice comparison for $\mu$.
More precisely, put $A_i^\circ:=R_{<i}X_{[m]\setminus\{i\}}A$,
so that $R_{\le i}K=R_iX_iA_i^\circ$. There is a density operator
$\tau_{i,h}^{\hX_iX_iA_i^\circ}$ with matching-copy support and
$X_i$-law supported on $\supp\mu_X$ such that
\begin{equation}
 \nu_{H_i=h}^{R_{\le i}K}=\mathcal C_\mu^A(\tau_{i,h}).
 \label{eq:ordered-restricted-seed}
\end{equation}
For Bob, interchange the players and reverse the coordinate order:
trace $R_{<i}$ and condition on $H_i=(T,\hX_{>i},Y_{<i})$.
Replace $R_{\le i},R_{<i}$ by $R_{\ge i},R_{>i}$ in the allocation
identity, and use $B_i^\circ:=R_{>i}Y_{[m]\setminus\{i\}}B$ for the
comparison form. The total is again $D(\rho\|\nu)-D(\rho^K\|\nu^K)$.
\end{enumerate}
\end{lemma}
\begin{proof}
For \ref{part:loss-classical-blocks}, apply the direct-sum identity
(Fact~\ref{fact:classical-kl-identities}) to both divergences.
Subtracting cancels the record divergence $D(p\|q)$ and gives the formula.

For \ref{part:ordered-telescope}, put
$\Delta_i:=D(\rho^{R_{\le i}K}\|\nu^{R_{\le i}K})$, with
$\Delta_0:=D(\rho^K\|\nu^K)$.
The registers $T,\hY_{<i}$ are classical in both Alice marginals,
and tracing $\hX_{>i}$ makes $X_{>i}$ classical.
Part~\ref{part:loss-classical-blocks} therefore gives
\begin{equation}
 \mathbb E_{H_i}\!\left[
 D(\rho_{H_i}^{R_{\le i}K}\|\nu_{H_i}^{R_{\le i}K})
 -D(\rho_{H_i}^{R_{<i}K}\|\nu_{H_i}^{R_{<i}K})\right]
 =\Delta_i-\Delta_{i-1}.
 \label{eq:main-exact-reference-increment}
\end{equation}
Summing gives $\Delta_m-\Delta_0$, proving
\eqref{eq:ordered-reference-telescope}. Bob uses the reverse order.

For preservation of the comparison form, the prior $\mu^{\otimes m}$ factors
across coordinates. Its comparison map therefore factors
as $\mathcal C_i\mathcal C_{-i}$, where $\mathcal C_i$ attaches
$\hY_i$ using $K_{y_i}^\mu$ on $X_i$, and $\mathcal C_{-i}$ attaches
the other references. Put $\zeta:=\mathcal C_{-i}(\tau)$, so
$\nu=\mathcal C_i(\zeta)$.

Let $\mathcal E_{i,h}$ trace $R_{>i}$ and select the block
$H_i=h$, retaining fixed record copies but leaving the block unnormalized.
This operation acts outside $\hX_iX_i\hY_i$, so it commutes with
$\mathcal C_i$. Consequently,
\[
 \mathcal E_{i,h}(\nu)
 =\mathcal C_i\!\left(\mathcal E_{i,h}(\zeta)\right).
\]
The map $\mathcal C_i$ preserves trace on the supported $X_i$ space.
Hence, for $p_\nu(h)>0$,
$\tau_{i,h}:=\mathcal E_{i,h}(\zeta)/p_\nu(h)$ is a density operator
and gives \eqref{eq:ordered-restricted-seed}. The restriction preserves
matching $\hX_iX_i$ support and the allowed $X_i$ symbols.
Bob's proof interchanges the players and reverses the order.
Finite parent divergence ensures positive comparison probability on
every target-supported record.
\end{proof}

\subsection{Information controlled by conditional relative entropy}
The conditional charge controls the variation of one player's marginal
with the other player's question. This will supply one of the two errors
in the local preparation proof.
\begin{lemma}[Information controlled by conditional relative entropy]
\label{lem:loss-cq-information}
\label{part:loss-cq}
Let $\rho^{\hY\hX XA}$ have matching-copy support and be classical on
$\hY$. Let $q^{XY}$ be its question law after tracing $\hX$, and let
$\nu_A$ be an Alice conditional comparison with finite full divergence.
Then, with $R=\hX\hY$,
\begin{equation}
 D(\rho^{RXA}\|\nu_A)-D(\rho^{XA}\|\nu_A^{XA})
 \ge I(\hY:A\mid X)_\rho+D(q^{XY}\|q_X\mu_{Y|X}).
 \label{eq:native-full-loss-cq-part}
\end{equation}
The inequality is an equality when both states are diagonal.
For a copied-input family, the left-hand side is $b_A$; interchanging
the players gives the corresponding bound for $b_B$.
\end{lemma}
\begin{proof}
Tracing $\hX$ makes $X$ classical. The defining purification gives
$\nu_A^{\hY XA}=\mu_{Y|X}\nu_A^{XA}$, with the conditional law attached
blockwise. Therefore
\begin{align*}
 &D(\rho^{RXA}\|\nu_A)-D(\rho^{XA}\|\nu_A^{XA})\\
 &\quad\ge D(\rho^{\hY XA}\|\mu_{Y|X}\nu_A^{XA})
             -D(\rho^{XA}\|\nu_A^{XA})
 &&\text{(data processing, Fact~\ref{fact:master-standard})}\\
 &\quad=I(\hY:A\mid X)_\rho+D(q^{XY}\|q_X\mu_{Y|X})
 &&\text{(chain rule, Fact~\ref{fact:classical-kl-identities}).}
\end{align*}
If both states are diagonal, tracing $\hX$ is reversible by copying
$X$ into $\hX$. Thus the data-processing step is an equality.
Bob's calculation interchanges $X,A$ with $Y,B$.
\end{proof}

\subsection{Matrix tools for marginal replacement}
To turn a conditional charge into a fidelity estimate, we will replace
the comparison's local marginal. The construction uses the following
trace inequality and variational formulas.

\begin{fact}[Four-matrix inequality and variational formulas {\cite{SBT}}]
\label{fact:four-matrix-trace}
Put $\beta_0(s):=\pi/[2(\cosh(\pi s)+1)]$, a probability density
on $\mathbb R$. For Hermitian matrices $H_1,\ldots,H_4$,
\begin{align*}
 \Tr e^{H_1+H_2+H_3+H_4}
 \le\int_{\mathbb R}\!\beta_0(s)\,
 \Tr\!\left[e^{H_1}e^{(1+is)H_2/2}e^{(1+is)H_3/2}
             e^{H_4}e^{(1-is)H_3/2}e^{(1-is)H_2/2}\right]ds.
\end{align*}
Write $D_e:=(\ln2)D$ and $D_{\mathbb M,e}:=(\ln2)D_{\mathbb M}$ for the
natural-log versions of the divergences in
Definition~\ref{def:information-quantities}. For a density operator $\alpha$
and positive definite $\zeta$, their variational characterizations are
\begin{align*}
 D_e(\alpha\|\zeta)
  &=\sup_{W>0}\{\Tr\alpha\ln W-\ln\Tr e^{\ln\zeta+\ln W}\},\\
 D_{\mathbb M,e}(\alpha\|\zeta)
  &=\sup_{W>0}\{\Tr\alpha\ln W-\ln\Tr\zeta W\}.
\end{align*}
The first formula also permits an unnormalized positive $\zeta$.
The measured-divergence formula extends to singular density operators
$\zeta$ and expresses $\DM$ as a supremum of continuous functions
of normalized states; hence $\DM$ is lower semicontinuous.
The trace inequality is Corollary~3.3, Eq.~(54), of Sutter, Berta,
and Tomamichel~\cite{SBT}. Their variational formulas
(48)--(49) have $1-z$ in place of $-\ln z$; optimizing the positive
scalar multiple of $W$ gives the equivalent logarithmic forms above.
\end{fact}

\section{Extraction objective and conditional information budget}
\label{sec:main-parent-allocation}
\label{sec:setup}
We now construct the coordinate targets and comparisons whose charges
control local preparation. Following JPY~\cite[Section~III]{JPY}, the
construction starts from the preselection and postselection states.

\subsection{The extraction objective}
We aim to turn success on unchecked coordinates, conditioned on checked
wins, into a legal strategy for one fresh game. The following notation
records the conditioning cost and the mean unchecked rejection probability.

\begin{definition}[Conditioned strategy, budget, and rejection]
\label{def:conditioned-strategy}
\normalfont
Fix a finite-dimensional strategy for $G^{\otimes n}$, with $G$ and
$L$ as in Definition~\ref{def:game}. Let $\mathsf P$ denote its measured
question-and-answer distribution, and let $W_i$ be success at coordinate
$i$. For $C\subsetneq[n]$, put $W_C:=\bigcap_{i\in C}W_i$ and
$p_C:=\Pr(W_C)$. Assume $p_C>0$ and define
\[
 J:=[n]\setminus C,\qquad m:=|J|,\qquad
 \mathsf Q:=\mathsf P(\,\cdot\mid W_C),
\]
\[
 B_C:=|C|L+\log(1/p_C),\qquad t:=B_C/m,\qquad
 A_C:=\frac1m\sum_{i\in J}\Pr(\neg W_i\mid W_C).
\]
For $i\in J$, let $q_i^{XY}$ be the $(X_i,Y_i)$ marginal of
$\mathsf Q$; write $q_i=q_i^{XY}$ when no other marginal is present.
Order $J$ increasingly, and for $i\in J$ write
\[
 J_{<i}:=\{j\in J:j<i\},\qquad
 J_{>i}:=\{j\in J:j>i\}.
\]
Set $J_{\le i}:=J_{<i}\cup\{i\}$ and
$J_{\ge i}:=J_{>i}\cup\{i\}$. In residual workspaces, write
$-i:=J\setminus\{i\}$. All probabilities refer to the fixed strategy.
\end{definition}

Our target is the extraction estimate \eqref{eq:cubic-branch-extraction}:
for $0\le t\le1$, the difference $\sqrt\varepsilon-\sqrt{A_C}$ is at most
$C_{\rm ext}t^{1/4}$ for a universal $C_{\rm ext}\ge1$. It forces $A_C\ge\varepsilon/4$ when
$t\le\varepsilon^2/(16C_{\rm ext}^4)$, the rejection guarantee needed for greedy
checking. To obtain it, we will approximately prepare the conditional
workspace of a uniformly chosen unchecked coordinate on fresh questions
from $\mu$. The information budget proved here supplies the input to
local preparation in Section~\ref{sec:local-preparation}.

\subsection{The two global states and their parent marginals}
Fix the conditioned strategy of Definition~\ref{def:conditioned-strategy}.
Purify the shared resource using Definition~\ref{def:canonical-purification}
and assign the purifying register to Alice. By the dilation
Fact~\ref{fact:isometric-dilation}, retain all outputs of the players'
local isometries. Write the resource as $\ket{\eta}^{A_0B_0}$ and the
checked answers as $Z_A,Z_B$. Retain coherent basis copies
$\widehat Z_A,\widehat Z_B$ of $Z_A,Z_B$, respectively, using the local
isometry $\ket z\ket0\mapsto\ket z\ket z$.
The outputs $E_A,E_B$ contain the unchecked answers and all remaining
workspaces. Copying preserves the measured strategy; register assignments
below are for the analysis and do not transmit these copies between players.

\begin{definition}[Preselection, postselection, and parent comparisons]
\label{def:two-global-states}
\label{def:checked-parent}
\label{def:repetition-comparisons}
\normalfont
Let $U_{\mathbf x}:A_0\to Z_A\widehat Z_AE_A$ and
$V_{\mathbf y}:B_0\to Z_B\widehat Z_BE_B$ include the checked-answer copies above.
The preselection and postselection vectors are
\begin{equation}
 \begin{aligned}
 \ket{\psi}
 &:=\sum_{\mathbf x,\mathbf y}\sqrt{\mu^{\otimes n}(\mathbf x,\mathbf y)}
 \ket{\mathbf x,\mathbf x,\mathbf y,\mathbf y}^{\hX^nX^n\hY^nY^n}
 (U_{\mathbf x}\otimes V_{\mathbf y})\ket{\eta},\\
 \ket{\phi}&:=\Pi_C\ket{\psi}/\sqrt{p_C},\qquad
 p_C=\bra{\psi}\Pi_C\ket{\psi}.
 \end{aligned}
 \label{eq:two-global-strategy-states}
\end{equation}
Here $\Pi_C$ is the diagonal projector onto the checked win event.
Put
\[
 H_0:=(X_C,Y_C),\qquad Z_C:=(Z_A,Z_B),\qquad
 H_{\rm chk}:=(H_0,Z_C).
\]
Their measured laws are $\mathsf P$ and $\mathsf Q=\mathsf P(\,\cdot\mid W_C)$.
The resource is not part of the checked transcript.
In the parent marginals below, the checked-question registers and all
four checked-answer registers are omitted. Put
\[
 R_J:=\hX_J\hY_J,\qquad K_A:=X_JE_A,\qquad K_B:=Y_JE_B,
 \qquad Q:=R_JK_AK_B.
\]
For a supported checked transcript $g=(h_0,z)$, the target and comparison
are the following normalized conditional marginals of the two states in
\eqref{eq:two-global-strategy-states}:
\begin{equation}
 \begin{array}{c@{\qquad}c}
 \text{target} & \text{comparison} \\[4pt]
 \phi_{H_{\rm chk}=g}^{Q} & \psi_{H_0=h_0}^{Q} \\[5pt]
 \phi_{H_{\rm chk}=g}^{R_JK_A} & \psi_{H_0=h_0}^{R_JK_A}.
 \end{array}
 \label{eq:two-state-parent-marginals}
\end{equation}
Bob's pair replaces $K_A$ by $K_B$. The target conditions on the
checked-answer value $z$; the comparison averages over the checked
answers, conditional on $h_0$. Both states omit all four checked-answer
registers and act on the same displayed residual systems.

\end{definition}
The budget proof verifies that the $R_JK_A$ and $R_JK_B$ marginals of
$\psi_{H_0=h_0}$ have the respective conditional-comparison forms.
The relative entropies compare the conditional marginals in
\eqref{eq:two-state-parent-marginals}, not the global pure states.

\subsection{Coordinate targets and conditional relative-entropy charges}
We now specify the coordinate families to be prepared and their information
charges. The construction below obtains both targets and comparisons
from the two states just defined.

\begin{definition}[Coordinate views and their comparisons]
\label{def:coordinate-losses}
\normalfont
For the postselection state $\ket{\phi}$ of
Definition~\ref{def:two-global-states}, a coordinate view at $i\in J$
partitions its registers as $R_iX_iY_iA_iB_i$, where $R_i:=\hX_i\hY_i$.
The register assignment is given in Definition~\ref{def:inherited-comparisons}.
Its normalized conditional workspace vectors satisfy
\begin{equation}
 \ket{\phi}
 =\sum_{x,y}\sqrt{q_i^{XY}(x,y)}\,
 \ket{xx,yy}^{\hX_iX_i\hY_iY_i}
 \ket{\phi_{X_i=x,Y_i=y}}^{A_iB_i}.
 \label{eq:joint-coordinate-target}
\end{equation}
The fixed copies of $X_i,Y_i$ are suppressed in the conditional vector. All records belong to $A_i$ or $B_i$, with the coordinate answers
in fixed positions. The view is \emph{consistent} when reading those
answers and the questions gives the $i$th marginal of $\mathsf Q$.

A coordinate comparison pair consists of
$\nu_{A,i}^{R_iX_iA_i}$ and $\nu_{B,i}^{R_iY_iB_i}$ of
Definition~\ref{def:conditional-comparisons}, with prior $\mu$.
When the full divergences are finite, let $b_{A,i},b_{B,i}$ be the local
conditional relative entropies of Definition~\ref{def:local-losses} for this coordinate view and
these comparisons. Put
\begin{equation}
 s_i:=b_{A,i}+b_{B,i},\qquad \kappa_i:=D(q_i^{XY}\|\mu).
 \label{eq:coordinate-loss-definition}
\end{equation}
\end{definition}

We next specify the partition and comparisons to which the ordered
chain rule applies.

\begin{definition}[Conditional coordinate states and comparisons]
\label{def:inherited-comparisons}
\normalfont
Define the full and comparison histories by
\begin{equation}
 H_i:=(H_{\rm chk},Y_{J_{<i}},X_{J_{>i}}),\qquad
 \overline H_i:=(H_0,Y_{J_{<i}},X_{J_{>i}}).
 \label{eq:common-measured-history}
\end{equation}
Alice's prefix trace makes $X_{J_{>i}}$ classical, and Bob's suffix
trace makes $Y_{J_{<i}}$ classical. These complementary traces expose
the same history and allow both conditional-comparison charges to telescope.
For a supported $h=(h_0,z,y_{J_{<i}},x_{J_{>i}})$, let $\bar h$ omit $z$.
The existing question and answer copies give two copies of $H_i$:
\[
 \begin{aligned}
 H_A&:=(X_C,\hY_C,Z_A,\widehat Z_B,\hY_{J_{<i}},X_{J_{>i}}),\\
 H_B&:=(\hX_C,Y_C,\widehat Z_A,Z_B,Y_{J_{<i}},\hX_{J_{>i}}).
 \end{aligned}
\]
The tuples have identical basis labels. Keep $R_i$ protected and put
\[
 \begin{aligned}
 A_i^\circ&:=\hX_{J_{<i}}X_{J_{<i}}E_A,& A_i&:=H_AA_i^\circ,\\
 B_i^\circ&:=\hY_{J_{>i}}Y_{J_{>i}}E_B,& B_i&:=H_BB_i^\circ.
 \end{aligned}
\]
These registers partition the global state as $R_iX_iY_iA_iB_i$.
This partition specifies the target for a new preparation problem.
It does not assert that the original players can transfer reference
registers: the extracted strategy prepares $A_iB_i$ from a fresh shared
resource using only $x_i$ and $y_i$. The selected reference $R_i$ is
used only in the conditional relative entropies.
\phantomsection\label{eq:main-pure-coordinate-target}
Conditioning on $H_i=h$ fixes both history copies, which are suppressed
in $\ket{\phi_{H_i=h}}^{R_iX_iY_iA_i^\circ B_i^\circ}$.
The target conditions on $H_i=h$, including the checked-answer value
$z$. The comparison conditions on $\overline H_i=\bar h$ and averages
over the checked answers. Both conditional blocks omit all
checked-answer registers.
For finite full divergences, define the local conditional relative
entropies for each supported history (Definition~\ref{def:local-losses}) directly from these marginals:
\begin{equation}
 b_{A,i,h}:=
 D\!\left(\phi_{H_i=h}^{R_iX_iA_i^\circ}
       \middle\|\psi_{\overline H_i=\bar h}^{R_iX_iA_i^\circ}\right)
 -D\!\left(\phi_{H_i=h}^{X_iA_i^\circ}
       \middle\|\psi_{\overline H_i=\bar h}^{X_iA_i^\circ}\right).
 \label{eq:parent-coordinate-marginal-identification}
\end{equation}

Define $b_{B,i,h}$ by replacing $X_iA_i^\circ$ with $Y_iB_i^\circ$.

After conditioning, fixed history factors can be removed or restored
without changing either divergence. This identifies the displayed marginals
with the restrictions of Lemma~\ref{lem:ordered-reference-loss}.

\phantomsection\label{eq:inherited-coordinate-seed}
The comparison blocks depend on $\bar h$, not $z$. Both conditional
states in \eqref{eq:parent-coordinate-marginal-identification} act on
$R_iX_iA_i^\circ$ (and on $X_iA_i^\circ$ after tracing $R_i$). The full
label $h$, including $z$, is restored in both recorded states below;
only the comparison's conditional block ignores $z$.
Proposition~\ref{prop:game-information-interface} verifies their
conditional-comparison form.

Let $q_i^H$ be the law of $H_i$ under $\mathsf Q$.
Tracing Bob's history copy gives the target below; define its comparison
with the same postselection history weights:
\begin{equation}
 \begin{aligned}
 \phi^{R_iX_iA_i}
   &=\sum_h q_i^H(h)\proj h^{H_A}\otimes
                    \phi_{H_i=h}^{R_iX_iA_i^\circ},\\
 \nu_{A,i}
   &:=\sum_h q_i^H(h)\proj h^{H_A}\otimes
                    \psi_{\overline H_i=\bar h}^{R_iX_iA_i^\circ}.
 \end{aligned}
 \label{eq:inherited-recorded-comparison}
\end{equation}
Bob's pair uses the replacements
$(H_A,X_i,A_i^\circ)\mapsto(H_B,Y_i,B_i^\circ)$.
The weights come from $\phi$ and the normalized comparison blocks from
$\psi$. Each block uses its own state's normalization. All weights and
blocks are fixed before the fresh questions arrive.
The retained-record identity gives the conditional relative entropies as averages under $q_i^H$.
\end{definition}

\Needspace{12\baselineskip}
\subsection{Conditioning and allocation of the charge}
\begin{proposition}[Full workspace budget for a repeated strategy]
\label{prop:game-information-interface}
For every conditioned strategy of Definition~\ref{def:conditioned-strategy},
the coordinate views of $\phi$ in
Definition~\ref{def:inherited-comparisons} are consistent targets of
Definition~\ref{def:coordinate-losses}. Use the comparisons assembled
from the prescribed conditional marginals in
\eqref{eq:inherited-recorded-comparison}.
Their full divergences are finite, and
\begin{equation}
 \frac1m\sum_{i\in J}s_i\le2t,\qquad
 \frac1m\sum_{i\in J}\kappa_i\le t.
 \label{eq:joint-information-budget}
\end{equation}
All records belong to the target workspaces. The targets, comparisons,
and their local spaces are fixed before the fresh questions arrive.
\end{proposition}
\begin{proof}
We first derive the budget by conditioning and telescoping, then verify
the required coordinate identities.

\emph{1. Bound the parent divergence.}
Use Lemma~\ref{lem:conditional-marginal-budget} on the measured
$H_0,Z_C$ registers of $\psi$, with $W=W_C$ and residual system $Q$.
On each winning transcript $g=(h_0,z)$,
$\psi_{H_{\rm chk}=g}^{Q}=\phi_{H_{\rm chk}=g}^{Q}$.
Since $\log|\mathcal Z_C|\le |C|L$, the lemma gives
\begin{equation}
 \mathbb E_{g\leftarrow\mathsf Q}
 D\!\left(\phi_{H_{\rm chk}=g}^{Q}\middle\|\psi_{H_0=h_0}^{Q}\right)
 \le B_C
 \label{eq:two-state-parent-budget}
\end{equation}
and finite divergence on every supported transcript. Data processing
(Fact~\ref{fact:master-standard}) gives the same bound for each player's
marginal.

\emph{2. Allocate the parent budget.}
Using the coordinate identities verified in Step 3, the retained-record
identity and ordered chain rule give the first equality below:
\begin{align}
 \sum_{i\in J}b_{A,i}
 &=\mathbb E_{H_{\rm chk}}\!\left[
 D\!\left(\phi_{H_{\rm chk}}^{R_JK_A}\middle\|\psi_{H_0}^{R_JK_A}\right)
 -D\!\left(\phi_{H_{\rm chk}}^{K_A}\middle\|\psi_{H_0}^{K_A}\right)
 \right]\notag\\
 &\le\mathbb E_{H_{\rm chk}}
 D\!\left(\phi_{H_{\rm chk}}^{Q}\middle\|\psi_{H_0}^{Q}\right)
 &&\text{(nonnegativity and data processing)}\notag\\
 &\le B_C
 &&\text{by \eqref{eq:two-state-parent-budget}.}
 \label{eq:ordered-parent-budget}
\end{align}
Bob's reverse-order application gives $\sum_{i\in J}b_{B,i}\le B_C$.
Consequently,
\begin{equation}
 \sum_{i\in J}(b_{A,i}+b_{B,i})\le2B_C.
 \label{eq:general-native-loss-allocation}
\end{equation}

The local isometries preserve the preselection question law
$\mu^{\otimes n}$, so the separate question-law bound follows from
postselection:
\begin{align*}
 \sum_{i\in J}D(q_i^{XY}\|\mu)
 &\le\log(1/p_C)
 &&\text{(Fact~\ref{fact:postselection})}\\
 &\le B_C
 &&\text{(definition of $B_C$).}
\end{align*}
Dividing these bounds by $m$ gives \eqref{eq:joint-information-budget}.

\emph{3. Verify the comparisons and coordinate identities.}
Tracing Bob's physical systems from $\psi$ eliminates $V_{\mathbf y}$.
For fixed $h_0=(x_C,y_C)$, put $F_A:=Z_A\widehat Z_AB_0$ and define
\begin{equation}
 \ket{\xi_{\mathbf x_J}}^{E_AF_A}
 :=(U_{(x_C,\mathbf x_J)}\otimes I_{B_0})\ket{\eta}.
 \label{eq:two-state-parent-seed}
\end{equation}
With $r=\mu_X^{\otimes m}$, these unit vectors give the defining
purification of $\psi_{H_0=h_0}^{R_JK_A}$: tracing $Y_JF_A$ also forgets
both copies of Alice's checked answer. Thus Lemma~\ref{lem:ordered-reference-loss}\ref{part:ordered-comparisons}
applies with $\nu=\psi_{H_0=h_0}^{R_JK_A}$ and $X^mA=X_JE_A$.
Bob's symmetric application uses $\psi_{H_0=h_0}^{R_JK_B}$. This proves
the comparison form for both comparison marginals in
\eqref{eq:parent-coordinate-marginal-identification}.

The partition in Definition~\ref{def:inherited-comparisons} makes the
coordinate target a view of the same postselection vector. Expanding its
two history copies gives
\begin{equation}
 \ket{\phi}
 =\sum_h\sqrt{q_i^H(h)}\,
 \ket{h,h}^{H_AH_B}
 \ket{\phi_{H_i=h}}^{R_iX_iY_iA_i^\circ B_i^\circ}.
 \label{eq:joint-coordinate-record-construction}
\end{equation}
The normalized conditional vectors and orthogonal history labels preserve
the question-and-answer probabilities. Reading the selected questions
and coordinate answers therefore gives the corresponding marginal of
$\mathsf Q$, establishing consistency.

The recorded comparisons retain the conditional form by linearity.
Write the comparison block as
$\psi_{\overline H_i=\bar h}^{R_iX_iA_i^\circ}
=\mathcal C_\mu^A(\tau_{i,\bar h})$. Then
\[
 \nu_{A,i}
 =\mathcal C_\mu^A\!\left(
   \sum_h q_i^H(h)\proj h^{H_A}\otimes\tau_{i,\bar h}\right),
\]
where the map acts trivially on $H_A$. The seed is a density operator
with the required matching-copy support; Bob is symmetric.
The two block decompositions in
\eqref{eq:inherited-recorded-comparison} have the same history weights $q_i^H$.
The parent bound ensures finite divergence on every supported block.
Applying the retained-record identity
(Lemma~\ref{lem:loss-chain-rule}\ref{part:loss-classical-blocks}) yields
\begin{equation}
 b_{A,i}=\mathbb E_{h\leftarrow q_i^H}b_{A,i,h},\qquad
 b_{B,i}=\mathbb E_{h\leftarrow q_i^H}b_{B,i,h}.
 \label{eq:joint-history-loss}
\end{equation}
Expectations use the postselection law throughout.

Fix $g=(h_0,z)$ and use the $R_JK_A$ and $R_JK_B$ marginals in
\eqref{eq:two-state-parent-marginals}.
In Alice's prefix marginal, $\hY_{J_{<i}}$ is classical and tracing
$R_{J_{>i}}$ makes $X_{J_{>i}}$ classical. Bob's suffix marginal has the
complementary classical record $(\hX_{J_{>i}},Y_{J_{<i}})$.
Matching copies identify both with the same history $H_i$ and target
probabilities. The conditional marginals are those of
Definition~\ref{def:inherited-comparisons}, after suppressing fixed
history copies. Each comparison is normalized by its own probability.

Thus Lemma~\ref{lem:ordered-reference-loss}\ref{part:ordered-telescope}
applies with $\rho=\phi_{H_{\rm chk}=g}^{R_JK_A}$ and
$\nu=\psi_{H_0=h_0}^{R_JK_A}$. Averaging over $g$ and using
\eqref{eq:joint-history-loss} proves the equality in
\eqref{eq:ordered-parent-budget}; Bob uses the reverse order.
\end{proof}

\section{Local preparation from conditional relative entropy}
\label{sec:local-preparation}
\label{sec:standalone-conditional-decoupling}
Proposition~\ref{prop:game-information-interface} bounds the mean coordinate
charges. We now prove the preparation estimate needed for extraction:
a total charge $s=b_A+b_B$ gives mean infidelity $O(\sqrt{s})$.
We derive this guarantee from nearby pure descriptions, then construct
them from two marginal errors: $b_A$ controls coherent alignment, and
$b_B$ controls the variation of Bob's marginal with $x$.

\begin{lemma}[Local preparation and pure descriptions]
\label{lem:smooth-preparation}
For a copied-input family with conditional comparisons and finite full
divergences (Definitions~\ref{def:copied-input-family},
\ref{def:conditional-comparisons}, and~\ref{def:local-losses}), put
$s:=b_A+b_B$.
\begin{enumerate}[label=(\roman*),ref=(\roman*),leftmargin=*,itemsep=5pt]
\item\label{part:preparation-sampling} \emph{Conditional preparation.}
The error of Definition~\ref{def:local-preparation-error} satisfies
\begin{equation}
 d_{\rm prep}(q,\psi)\le\min\{1,C_{\rm prep}\sqrt s\},
 \qquad C_{\rm prep}=100C_{\rm sm}.
 \label{eq:smooth-native-preparation}
\end{equation}
Every positive slack has a finite-resource witness with designated
output $AB$.
\item\label{part:preparation-descriptions} \emph{Pure descriptions and local corrections.}
After common finite pure padding, there are bipartite unit vectors
$\ket{u_x},\ket{v_y}$ on $A_{\rm s}F$, local isometries $J_x:A_{\rm s}\to A$,
and unitaries $W_y:B\to F$ such that
\begin{equation}
 \boxed{\begin{aligned}
 \mathbb E_q\|\ket{u_x}-\ket{v_y}\|_2^2
   &\le4(\ln2)(b_A+b_B),\\
 \mathbb E_qP^2\!\left((J_x\otimes W_y^\dagger)u_x(J_x^\dagger\otimes W_y),\psi_{xy}\right)
   &\le(\ln2)b_A.
 \end{aligned}}
 \label{eq:pure-preparation-interface}
\end{equation}
Alice selects $\ket{u_x},J_x$ using only $x$; Bob selects $\ket{v_y},W_y$ using
only $y$. All descriptions, corrections, resources, and preparation maps
are fixed from the family and comparisons before the questions arrive.
\end{enumerate}
\end{lemma}
\begin{proof}
We first obtain \ref{part:preparation-sampling} from
\ref{part:preparation-descriptions}, then construct the descriptions.

\emph{1. Derive preparation from the pure descriptions.}
Assume \ref{part:preparation-descriptions} and put
$\mathcal D:=\mathbb E_q\|\ket{u_x}-\ket{v_y}\|_2^2\le4(\ln2)s$.
Fix $\epsilon_{\rm sm}>0$. Song's sampler
(Fact~\ref{fact:song-smooth-sampling}) gives finite resources producing
$\sigma_{xy}$ with
\[
 \mathbb E_qP^2(\sigma_{xy},u_x)
 \le4C_{\rm sm}\sqrt{\mathcal D}+\epsilon_{\rm sm}.
\]
Let $\sigma'_{xy}$ be the output after Alice applies $J_x$ and Bob
applies $W_y^\dagger$. The description estimate and isometry invariance give
\begin{align*}
 \mathbb E_qP^2(\sigma'_{xy},\psi_{xy})
 &\le2\mathbb E_qP^2(\sigma_{xy},u_x)+2(\ln2)b_A
 &&\text{(squared triangle inequality and alignment)}\\
 &\le8C_{\rm sm}\sqrt{\mathcal D}+2(\ln2)b_A+2\epsilon_{\rm sm}
 &&\text{(sampling bound)}\\
 &\le16\sqrt{\ln2}\,C_{\rm sm}\sqrt s+2(\ln2)s+2\epsilon_{\rm sm}
 &&\text{($\mathcal D\le4(\ln2)s$, $b_A\le s$).}
\end{align*}
For $s\le1$, the last expression is at most
$C_{\rm prep}\sqrt s+2\epsilon_{\rm sm}$.
Taking the infimum over finite preparations and letting
$\epsilon_{\rm sm}\downarrow0$ gives the asserted bound on $d_{\rm prep}$.
The bound $P^2\le1$ supplies the cap and covers $s>1$.
Every positive slack is realized with finite resources; discarding
auxiliary environments and padding preserves the bound on $AB$.

\emph{2. Construct the descriptions from two marginal errors.}
We will construct unit vectors $\ket{\chi_x}^{AF}$ and unitaries
$W_y:B\to F$ with alignment error
\[
 e_A:=\mathbb E_qP^2\!\left(
 (I_A\otimes W_y)\psi_{xy}(I_A\otimes W_y^\dagger),\chi_x\right)
 \le(\ln2)b_A.
\]
Bob's conditional marginal is
$\psi_y^B:=\sum_xq(x\mid y)\psi_{xy}^B$.
Its mean variation is bounded by
\begin{align}
 e_B:=\mathbb E_qP^2(\psi_{xy}^B,\psi_y^B)
 &\le(\ln2)\mathbb E_qD(\psi_{xy}^B\|\psi_y^B)
 &&\text{(Fact~\ref{fact:divergence-fidelity})}\notag\\
 &=(\ln2)I(X:B\mid Y)_\psi
 &&\text{(Fact~\ref{fact:classical-kl-identities})}\notag\\
 &\le(\ln2)b_B
 &&\text{(Lemma~\ref{lem:loss-cq-information}).}
 \label{eq:direct-bob-information}
\end{align}
The information is evaluated on the classical-question marginal
$\psi^{XYB}$.

Assuming the alignment bound, take canonical purifications on $A_{\rm s}F$:
\begin{equation}
 \ket{u_x}:=\vket{\sqrt{\chi_x^F}},\qquad
 \ket{v_y}:=\vket{W_y\sqrt{\psi_y^B}W_y^\dagger}.
 \label{eq:pure-local-descriptions}
\end{equation}
Purification equivalence (Fact~\ref{fact:local-transitions}) supplies
Alice's isometry $J_x$ with
\[
 (J_x\otimes I_F)\ket{u_x}=\ket{\chi_x}.
\]
Thus the mean error after the corrections is $e_A$.
For the vector discrepancy, insert $W_y\psi_{xy}^BW_y^\dagger$
between the two marginals and use the squared triangle inequality and
contraction under tracing $A$:
\begin{align}
 \mathcal D
 &\le2\mathbb E_qP^2(\chi_x^F,W_y\psi_y^BW_y^\dagger)
 &&\text{(Fact~\ref{fact:canonical-sampling-discrepancy})}\notag\\
 &\le4(e_A+e_B)
 &&\text{(squared triangle inequality and contraction)}\notag\\
 &\le4(\ln2)s
 &&\text{(the two error bounds).}
 \label{eq:direct-posterior-descriptions}
\end{align}
This proves \ref{part:preparation-descriptions}, subject to constructing
$\ket{\chi_x}$ and $W_y$ with the asserted alignment bound.

\emph{3. Obtain coherent alignment by marginal replacement.}
Lemma~\ref{lem:retained-marginal-lift}\ref{part:replacement-conditional},
proved below, and Fact~\ref{fact:divergence-fidelity} give an Alice
conditional comparison $\widetilde\nu_A$ such that
\[
 \widetilde\nu_A^{XA}=\psi^{XA},\qquad
 F_0:=F(\psi^{RXA},\widetilde\nu_A)\ge2^{-b_A/2}.
\]
By Lemma~\ref{lem:conditional-comparison-forms}, choose a purification
$\ket\chi^{RXYAF}$ of $\widetilde\nu_A$ whose conditional $AF$ vector
$\ket{\chi_x}$ depends only on $x$. The exact $XA$ marginal fixes its
question law as $r(x,y):=q_X(x)\mu(y\mid x)$.
Both reduced states are classical on $\hY$. On every $q_Y$-supported
block, $r_Y(y)>0$. Uhlmann's theorem gives a unitary $W_y:B\to F$
whose normalized block overlap is real and nonnegative. Each $W_y$
acts on the whole coherent $x$ block at fixed $y$.
Writing
\[
 a_{xy}:=\bra{\chi_x}(I_A\otimes W_y)\ket{\psi_{xy}},
\]
block fidelity gives the nonnegative sum below; individual $a_{xy}$
may be complex:
\begin{equation}
 F_0=\sum_{x,y}\sqrt{q(x,y)r(x,y)}\,a_{xy}
 \qquad\text{(Facts~\ref{fact:local-transitions} and~\ref{fact:block-fidelity}).}
 \label{eq:direct-coherent-overlap}
\end{equation}
Although the overlap uses both laws $q$ and $r$, Cauchy--Schwarz
gives an error averaged under the target law $q$:
\begin{align}
 e_A&=1-\sum_{x,y}q(x,y)|a_{xy}|^2
 &&\text{(pure-state infidelity)}\notag\\
 &\le1-F_0^2
 &&\text{(Cauchy--Schwarz and $\sum_{x,y}r(x,y)=1$)}\notag\\
 &\le1-2^{-b_A}
 &&\text{($F_0\ge2^{-b_A/2}$)}\notag\\
 &\le(\ln2)b_A
 &&\text{($1-e^{-u}\le u$).}
 \label{eq:direct-target-infidelity}
\end{align}
This completes the alignment bound and the proof.
\end{proof}

\subsection{Marginal replacement for conditional comparisons}
\label{sec:main-marginal-replacement}
We now prove the replacement estimate used in Step 3 of
Lemma~\ref{lem:smooth-preparation}. A small conditional relative entropy need not make the original
comparison close to the target, because the retained-marginal divergence
has been subtracted. Matching that marginal turns this quantity into a
measured-divergence bound while preserving the conditional
comparison form.

\Needspace{6\baselineskip}
\begin{lemma}[Constructive marginal replacement]
\label{lem:retained-marginal-lift}

\label{lem:comparison-marginal-replacement}
\leavevmode
\begin{enumerate}[label=(\roman*),ref=(\roman*),leftmargin=*,itemsep=5pt]
\item\label{part:replacement-general} \emph{Explicit replacement.}
Let $\rho^{RK}$ be a density operator and $\nu^{RK}$ a nonzero positive
operator with $D(\rho\|\nu)<\infty$; $\nu$ need not be normalized.
With $\beta_0$ from Fact~\ref{fact:four-matrix-trace}, define
\begin{equation}
 \begin{aligned}
 L_s&:=(\rho^K)^{(1+is)/2}(\nu^K)^{-(1+is)/2},\\
 \widetilde\nu&:=\int_{\mathbb R}\beta_0(s)
 (I_R\otimes L_s)\nu(I_R\otimes L_s)^\dagger\,ds.
 \end{aligned}
 \label{eq:marginal-replacement-construction}
\end{equation}
Inverse powers act on $\supp\nu^K$ and vanish on its kernel; positive
powers of $\rho^K$ vanish on $\ker\rho^K$.
Then $\widetilde\nu$ is a density operator satisfying
\begin{equation}
 \widetilde\nu^K=\rho^K,\qquad
 D_{\mathbb M}(\rho\|\widetilde\nu)
 \le D(\rho\|\nu)-D(\rho^K\|\nu^K).
 \label{eq:retained-marginal-lift}
\end{equation}
The replacement depends on the target $\rho$ only through $\rho^K$.
\item\label{part:replacement-conditional} \emph{Preservation of conditional comparisons.}
For the states in part~\ref{part:replacement-general}, suppose that
$\nu$ is an Alice conditional comparison for $\mu$
(Definition~\ref{def:conditional-comparisons}), $R=\hX\hY$, $K=XA$,
and $\rho$ has matching-copy support. The same $\widetilde\nu$ is an
Alice conditional comparison for $\mu$. In particular,
\begin{equation}
 \widetilde\nu^{XA}=\rho^{XA},\qquad
 D_{\mathbb M}(\rho\|\widetilde\nu)\le D(\rho\|\nu)-D(\rho^{XA}\|\nu^{XA}).
 \label{eq:native-marginal-lift}
\end{equation}
\end{enumerate}
Neither part requires a full-support or smallest-eigenvalue hypothesis.
\end{lemma}
\begin{proof}
The support assumption also gives $\supp\rho^K\subseteq\supp\nu^K$.
We first deduce the conditional statement~\ref{part:replacement-conditional}
from the general estimate~\ref{part:replacement-general}, then prove that estimate.

\phantomsection\label{sec:comparison-preservation}
\emph{\ref{part:replacement-conditional} Preservation of conditional comparisons.}
By Lemma~\ref{lem:conditional-comparison-forms}, write
$\nu=\mathcal C_\mu^A(\tau_A)$, with $\tau_A^{XA}=\nu^{XA}$.
Both $XA$ marginals are $X$-block diagonal, so
$L_s=\sum_x\proj x^X\otimes L_{s,x}$ commutes with every $K_y^\mu$.
Define
\[
 \widetilde\tau_A:=\int_{\mathbb R}\beta_0(s)
 (I_{\hX}\otimes L_s)\tau_A(I_{\hX}\otimes L_s)^\dagger\,ds.
\]
Each integrand is positive, has the required matching-copy support,
and has trace $\Tr L_s\nu^{XA}L_s^\dagger=\Tr\rho^{XA}=1$.
Thus $\widetilde\tau_A$ is a density operator on the same finite space,
with $X$-law supported on $\supp\mu_X$. Moreover,
\begin{align*}
 \mathcal C_\mu^A(\widetilde\tau_A)
 &=\int_{\mathbb R}\beta_0(s)(I_R\otimes L_s)
      \mathcal C_\mu^A(\tau_A)(I_R\otimes L_s)^\dagger\,ds
 &&\text{(linearity and commutation)}\\
 &=\widetilde\nu
 &&\text{(\eqref{eq:marginal-replacement-construction}).}
\end{align*}
Lemma~\ref{lem:conditional-comparison-forms} now gives the comparison
form. The exact marginal and divergence bound follow from
part~\ref{part:replacement-general}.

\emph{\ref{part:replacement-general} Explicit replacement.}
Scaling $\nu$ changes both divergences by the same constant and leaves
$\widetilde\nu$ unchanged. Normalize $\nu$. We first prove the estimate
for positive definite $\rho$ and $\nu$, then extend it by regularization
under the standing support assumption. Since $L_s\nu^KL_s^\dagger=\rho^K$,
each integrand has trace one and the required marginal.

Use natural logarithms and write
$b_e=D_e(\rho\|\nu)-D_e(\rho^K\|\nu^K)$.
For any operator $H$ on $K$,
$\Tr\rho(I_R\otimes H)=\Tr\rho^K H$. Thus the two marginal logarithms
can be evaluated against $\rho$; their extensions by $I_R$ are implicit.
Expanding the relative entropies gives
\begin{align*}
 b_e
 &=\Tr\rho\bigl(\ln\rho-\ln\nu-\ln\rho^K+\ln\nu^K\bigr)
 &&\text{(partial-trace identity)}\\
 &=D_e\!\left(\rho\middle\|
       e^{\ln\nu+\ln\rho^K-\ln\nu^K}\right)
 &&\text{(definition of $D_e$).}
\end{align*}
The second argument is positive but need not be normalized.
For every $W>0$, Fact~\ref{fact:four-matrix-trace}, with ordered matrices
$\ln W,\ln\rho^K,-\ln\nu^K,\ln\nu$, now gives
\begin{align*}
 b_e&\ge\Tr\rho\ln W-\ln\Tr
       e^{\ln W+\ln\rho^K-\ln\nu^K+\ln\nu}
 &&\text{(variational formula)}\\
 &\ge\Tr\rho\ln W-\ln\Tr W\widetilde\nu
 &&\text{(four-matrix inequality).}
\end{align*}
Taking the supremum over $W$ and dividing by $\ln2$ proves \ref{part:replacement-general} in this case.

\phantomsection\label{sec:marginal-replacement-supports}
We now remove the temporary full-rank assumption. With $\nu$
normalized as above, put
\[
 \rho_\epsilon=(1-\epsilon)\rho+\epsilon I/\dim(RK),
\]
with $\nu_\epsilon$ defined analogously. Continuity of $t\ln t$,
convergence of $\ln\nu_\epsilon$ on $\supp\nu$, and
$\epsilon\ln\epsilon\to0$ give convergence of the full divergence;
the same argument applies to the marginals.

Let $Q$ project onto $\ker\nu^K$. The support assumption gives
\[
 \rho_\epsilon^KQ=\nu_\epsilon^KQ=\frac{\epsilon}{\dim K}Q,
 \qquad L_{s,\epsilon}Q=Q,
\]
where $L_{s,\epsilon}$ is defined using the regularized marginals.
The projection $Q$ reduces both marginals, and positivity forces $\nu$
to vanish on $R\otimes\ker\nu^K$. Thus the regularized integrand has
$Q$-block $\epsilon(I_R\otimes Q)/\dim(RK)$ and no cross terms with its
orthogonal complement. This block vanishes as $\epsilon\downarrow0$.
On $\supp\nu^K$, inverse powers remain bounded and positive powers
converge, including on $\ker\rho^K$. For fixed $s$, the integrand therefore
converges in trace norm to the supported formula.
Each integrand has trace norm one, so dominated convergence against
$\beta_0(s)\,ds$ yields the integrated state and its exact marginal.
Lower semicontinuity of $D_{\mathbb M}$
(Fact~\ref{fact:four-matrix-trace}) gives the limiting bound.

\end{proof}

\section{Completing extraction and repetition}
\label{sec:game}
Combining the information budget with local preparation gives the
single-game strategy needed for the extraction objective of
Section~\ref{sec:main-parent-allocation}. Its rejection bound yields
all-win repetition by greedy checking and threshold repetition by Rao's
concentration argument.

\subsection{Single-copy extraction under fresh questions}
The preparation bound is used only on the designated game workspace.
We compare each branch's rejection probability before changing its
question law.

\begin{proposition}[Single-copy extraction from smooth preparation]
\label{prop:cubic-extraction}
For consistent coordinate targets of
Definition~\ref{def:coordinate-losses}, assume
$m^{-1}\sum_i s_i\le2t$ and $m^{-1}\sum_i\kappa_i\le t$.
With $A_C$ from Definition~\ref{def:conditioned-strategy} and
$C_{\rm prep}$ from Lemma~\ref{lem:smooth-preparation},
\begin{equation}
 \sqrt\varepsilon\le\sqrt{A_C}
 +\sqrt{C_{\rm prep}}\,(2t)^{1/4}+\sqrt{(\ln2)t}.
 \label{eq:cubic-branch-extraction-exact}
\end{equation}
In particular, for $0\le t\le1$,
\begin{equation}
 \boxed{\sqrt\varepsilon\le\sqrt{A_C}+C_{\rm ext} t^{1/4},\qquad
 C_{\rm ext}:=2^{1/4}\sqrt{C_{\rm prep}}+\sqrt{\ln2}.}
 \label{eq:cubic-branch-extraction}
\end{equation}
The constant $C_{\rm ext}\ge1$ is universal.
\end{proposition}

\begingroup
\mdfsetup{nobreak=true}
\begin{procedurebox}{The extracted single-game strategy}
\label{proc:cubic-single-copy}
Fix the repeated strategy and checked set of
Definition~\ref{def:conditioned-strategy}, and consistent coordinate
targets satisfying the assumptions of Proposition~\ref{prop:cubic-extraction}.
Use the preparation maps of
Lemma~\ref{lem:smooth-preparation}\ref{part:preparation-sampling}.
\begin{enumerate}[leftmargin=*,itemsep=3pt]
\item Before receiving the questions, the players share a uniformly chosen
unchecked coordinate $i$ and a preparation resource for that coordinate.
Their joint shared state is independent of the fresh questions.
\item Alice uses $(i,x)$ and Bob uses $(i,y)$ to select their local
preparation maps, including the canonical corrections. They send no message.
\item They read the designated coordinate answers and send them to
the referee. All target records belong to the prepared workspaces.
\end{enumerate}
\end{procedurebox}
\endgroup

\begin{proof}
Fix $\epsilon_{\rm prep}>0$ and choose each coordinate's preparation resource and local
maps with this slack. Sharing the coordinate $i$ uniformly together
with its resource realizes Procedure~\ref{proc:cubic-single-copy}.
Its shared state is independent of the fresh questions, so its
rejection probability is at least $\varepsilon$.

Apply Lemma~\ref{lem:square-root-rejection} with label $z=(i,x,y)$,
target $\phi_{X_i=x,Y_i=y}^{A_iB_i}$, prepared state $\sigma_{i,x,y}$,
and the rejecting effect that reads the designated answers and tests
the game predicate. The target and fresh-question laws are
\[
 q(z):=\frac1m q_i^{XY}(x,y),\qquad p(z):=\frac1m\mu(x,y).
\]
Consistency gives $a_q=A_C$. The two error bounds are
\begin{align*}
 e_q&=\frac1m\sum_{i\in J}\mathbb E_{q_i}
 P^2(\sigma_{i,x,y},\phi_{X_i=x,Y_i=y}^{A_iB_i})\\
 &\le\frac{C_{\rm prep}}m\sum_{i\in J}\sqrt{s_i}+\epsilon_{\rm prep}
 &&\text{(Lemma~\ref{lem:smooth-preparation}\ref{part:preparation-sampling})}\\
 &\le C_{\rm prep}\sqrt{2t}+\epsilon_{\rm prep}
 &&\text{(Jensen and the charge budget)},\\
 D(q\|p)&=\frac1m\sum_{i\in J}D(q_i^{XY}\|\mu)
 &&\text{(direct-sum identity)}\\
 &\le t
 &&\text{(question-law budget).}
\end{align*}
Let $a_p$ denote the extracted strategy's rejection probability
under fresh questions. Since $a_p\ge\varepsilon$ and $a_q=A_C$,
Lemma~\ref{lem:square-root-rejection} gives
\begin{align*}
 \sqrt\varepsilon
 &\le\sqrt{a_p}
 &&\text{(single-game value)}\\
 &\le\sqrt{A_C}+\sqrt{e_q}+\sqrt{(\ln2)D(q\|p)}
 &&\text{(rejection comparison)}\\
 &\le\sqrt{A_C}
 +\sqrt{C_{\rm prep}\sqrt{2t}+\epsilon_{\rm prep}}+\sqrt{(\ln2)t}
 &&\text{(the two bounds above).}
\end{align*}
Let $\epsilon_{\rm prep}\downarrow0$ to get \eqref{eq:cubic-branch-extraction-exact}.
Every positive slack uses finite resources. For $t\le1$,
$\sqrt t\le t^{1/4}$ gives \eqref{eq:cubic-branch-extraction}.
\end{proof}

\subsection{All-win repetition}
\label{sec:main-repetition-proof}
We apply the extraction estimate of Proposition~\ref{prop:cubic-extraction}
to prove the main theorem by greedy checking.

\begin{proof}[Proof of Theorem~\ref{thm:smooth-repetition}]
For every supported checked set, Proposition~\ref{prop:game-information-interface}
supplies targets satisfying the extraction hypotheses. We first derive
a rejection threshold $A_C\ge d_0$ when $t\le t_0$, then apply greedy
checking.

\emph{1. Obtain a rejection threshold.}
With the universal $C_{\rm ext}\ge1$ of Proposition~\ref{prop:cubic-extraction}, choose
\[
 t_0:=\frac{\varepsilon^2}{16C_{\rm ext}^4},\qquad
 d_0:=\frac{\varepsilon}{4},\qquad
 \ell_0:=-\log(1-d_0).
\]
If $t\le t_0$, the extraction estimate gives
\begin{align}
 \sqrt{A_C}&\ge\sqrt\varepsilon-C_{\rm ext}t^{1/4}
 &&\text{(extraction)}\notag\\
 &\ge\tfrac12\sqrt\varepsilon
 &&\text{($t\le t_0$).}
 \label{eq:cubic-threshold-extraction}
\end{align}
Hence $A_C\ge d_0$, so some unchecked coordinate rejects with probability
at least $d_0$.
Also $\ell_0+t_0\le\log(4/3)+1/16<1\le L$.

\emph{2. Greedy checking.}
\phantomsection\label{sec:greedy-counting}\label{lem:greedy-threshold}
We use the conditioning argument of Holenstein~\cite[Section~7]{Holenstein}.
Start with $C=\varnothing$. While $C\subsetneq[n]$, $p_C>0$, and
$B_C/(n-|C|)\le t_0$, add an unchecked coordinate whose conditional
rejection probability is at least $d_0$, as supplied by Step 1.
The process stops after at most $n$ checks.

If $p_C=0$ at termination, then $p_{[n]}=0$ and the bound is immediate.
Otherwise, put $k:=|C|$ and $\lambda:=\log(1/p_C)$.
Each check reduces the success probability by a factor at most $1-d_0$, so
\[
 \lambda\ge k\ell_0.
\]
At termination, either the budget test has failed or $k=n$. In both cases,
$\lambda+k(L+t_0)\ge t_0n$.
Using $k\le\lambda/\ell_0$ therefore gives
\begin{align*}
 t_0n&\le\lambda+k(L+t_0)
 &&\text{(termination)}\\
 &\le\lambda\left(1+\frac{L+t_0}{\ell_0}\right)
 &&\text{($k\le\lambda/\ell_0$).}
\end{align*}
Rearranging gives $\lambda\ge\ell_0t_0n/(L+\ell_0+t_0)$.
Since $p_{[n]}\le p_C=2^{-\lambda}$,
\begin{align}
 p_{[n]}&\le
 \exp\!\left[-\frac{(\ln2)\ell_0t_0n}{L+\ell_0+t_0}\right]\notag\\
 &\le\exp[-d_0t_0n/(2L)]
 &&\text{($\ell_0+t_0\le L$, $(\ln2)\ell_0\ge d_0$).}
 \label{eq:greedy-threshold}
\end{align}
Substituting the thresholds from Step 1 and taking the supremum over
the fixed repeated strategies yields
\begin{equation}
 \omega^*(G^{\otimes n})\le
 \exp[-c_{\rm sm}\varepsilon^3n/L],\qquad
 c_{\rm sm}=\frac1{128C_{\rm ext}^4}.
 \label{eq:smooth-native-repetition-constant}
\end{equation}
\end{proof}

\subsection{Threshold parallel repetition}
The extraction estimate controls average unchecked success. Applying
Rao's concentration bound (Fact~\ref{fact:rao-concentration}) gives the
following extension.

\begin{corollary}[Threshold parallel repetition]
\label{cor:threshold-repetition}
Let $G$ be a finite two-player game with
$\omega^*(G)=1-\varepsilon<1$ and
$L:=\max\{1,\log(|\cA||\cB|)\}$. There is a universal constant
$c_{\rm th}>0$ such that, for every $0<\delta\le\varepsilon$, every $n\ge1$,
and every finite-dimensional strategy for $G^{\otimes n}$, its number $N$
of won coordinates satisfies
\begin{equation}
 \Pr\!\left[N\ge\left\lceil(1-\varepsilon+\delta)n\right\rceil\right]
 \le\exp\!\left[-c_{\rm th}\frac{\delta^5n}{\varepsilon^2L}\right].
 \label{eq:threshold-repetition}
\end{equation}
\end{corollary}
\begin{proof}

Fix the strategy and put $v:=1-\varepsilon$. With $C_{\rm ext}$ from
Proposition~\ref{prop:cubic-extraction}, set
\[
  t_{\rm th}:=\left(\frac{\delta}{8C_{\rm ext}\sqrt\varepsilon}\right)^4
 =\frac{\delta^4}{8^4C_{\rm ext}^4\varepsilon^2}\le\frac{\delta}{4}\le\frac14.
\]
For every supported proper checked set $C$ with $t\le t_{\rm th}$,  the
extraction estimate and the workspace budget give
\begin{align*}
 \sqrt{A_C}&\ge\sqrt\varepsilon-C_{\rm ext}t^{1/4}
 &&\text{(extraction)}\\
 &\ge\sqrt\varepsilon-\frac{\delta}{8\sqrt\varepsilon}
 &&\text{($t\le t_{\rm th}$).}
\end{align*}
The right-hand side is positive since $\delta\le\varepsilon$.
Squaring yields $A_C\ge\varepsilon-\delta/4$. Thus
$1-A_C\le v+\delta/4$ for every such $C$, which is the hypothesis of
Fact~\ref{fact:rao-concentration}. Consequently,
\[
 \Pr\!\left[N\ge\left\lceil(v+\delta)n\right\rceil\right]
 \le K_{\rm R}\exp\!\left[-\frac{c_{\rm R}}{8^4C_{\rm ext}^4}
 \frac{\delta^5n}{\varepsilon^2L}\right].
\]

To remove $K_{\rm R}$, let $p$ be the threshold probability of the fixed
strategy and run $r$ independent blocks of it. If every block wins at
least $\lceil(v+\delta)n\rceil$ coordinates, their total is at least
$\lceil(v+\delta)rn\rceil$. Applying the preceding bound to this finite-dimensional
strategy on $rn$ coordinates gives
\[
 p^r\le K_{\rm R}\exp\!\left[
 -\frac{c_{\rm R}}{8^4C_{\rm ext}^4}
  \frac{\delta^5rn}{\varepsilon^2L}\right].
\]
Take $r$th roots and let $r\to\infty$. The choice
$c_{\rm th}:=c_{\rm R}/(8^4C_{\rm ext}^4)$ is universal.
\end{proof}

\section*{Acknowledgements}
This project is supported by the National Research Foundation, Singapore through the National Quantum Office, hosted in A*STAR, under its Centre for Quantum Technologies Funding Initiative (S24Q2D0009) and its Advanced Quantum Algorithms and Solutions Funding Initiative (S25Q9DA001 and S25Q9DA002).

The author used ChatGPT to assist with mathematical exploration,
checking calculations, revising the exposition, and preparing the LaTeX source. Responsibility for
the results and their presentation rests with the author.

\Needspace{8\baselineskip}
\appendix
\section{Product-question specialization}
\label{sec:main-benchmark-checks}
We verify the product-question identity used in the introduction.
\begin{lemma}[Product-prior identity]
\label{lem:benchmark-loss-identities}
For a copied-input family with $\mu=\mu_X\otimes\mu_Y$,
\begin{equation}
 \inf_{\nu_A} b_A
 =I(\hY:\hX XA)_\psi+D(q_Y\|\mu_Y),
 \label{eq:product-native-loss-minimum}
\end{equation}
where the infimum is over Alice conditional comparisons with finite full
divergence. Bob has the symmetric identity. The posterior $q$ need not
be product; this optimization identifies the quantity and does not
change the fixed comparisons used in allocation.
\end{lemma}
\begin{proof}
The defining purification gives
$\nu_A=\mu_Y^{\hY}\otimes\nu_A^{\hX XA}$.
For a bipartite state $\rho^{UV}$ and a product reference
$\alpha^U\otimes\beta^V$ with finite divergence, expanding the logarithms gives
\[
 D(\rho^{UV}\|\alpha^U\otimes\beta^V)
 =I(U:V)_\rho+D(\rho^U\|\alpha^U)+D(\rho^V\|\beta^V).
\]
Apply this with $U=\hY$, $V=\hX XA$, and $\rho=\psi^{RXA}$;
the $\hY$ marginal has law $q_Y$. Subtracting the retained-marginal
divergence yields
\begin{align*}
 b_A
 &=I(\hY:\hX XA)_\psi+D(q_Y\|\mu_Y)\\
 &\quad+D(\psi^{\hX XA}\|\nu_A^{\hX XA})
        -D(\psi^{XA}\|\nu_A^{XA})\\
 &\ge I(\hY:\hX XA)_\psi+D(q_Y\|\mu_Y)
 &&\text{(data processing, Fact~\ref{fact:master-standard}).}
\end{align*}
The conditional comparison $\nu_A=\mu_Y^{\hY}\otimes\psi^{\hX XA}$
achieves equality. Bob is symmetric. The product identity
recovers JPY's coherent opposite-question charge, with the marginal
question divergence explicit~\cite{JPY}.
\phantomsection\label{sec:product-averaging}
\end{proof}

\section{Parameter correspondence for Rao's concentration bound}
\label{sec:rao-parameters}
We verify the conditional-success form of Rao's bound in
Fact~\ref{fact:rao-concentration}, including the constants and rounding.
Put
\[
 a:=v+\delta/4,\qquad \lambda:=\log(1/a),\qquad
 D_0:=L+ t_{\rm th}+\lambda,\qquad k:=\left\lfloor\frac{ t_{\rm th} n}{D_0}\right\rfloor.
\]
Here $0<a<1$ and $k\le\delta n/4$. For $j=|C|<k$ and $p_C>a^k$,
\begin{align*}
 \frac{jL+\log(1/p_C)}{n-j}
 &\le\frac{k(L+\lambda)}{n-k}
 &&\text{($j<k$ and $p_C>a^k$)}\\
 &\le t_{\rm th}
 &&\text{($k\le t_{\rm th} n/D_0$).}
\end{align*}
Thus the conditional success is at most $a$ before the probability
drops to or below the fixed cutoff $a^k$. We spell out the step from Rao's
Lemma~31 that uses this hypothesis. For $k\ge1$, generate
$S_0=\varnothing,\ldots,S_k$ by successively adjoining a uniformly random
unused coordinate, independently of the game outcome, and put
\[
 F_j:=\mathbb E\!\left[p_{S_j}\mathbf1_{\{p_{S_j}>a^k\}}\right].
\]
If $p_{S_{j-1}}>a^k$, the preceding conditional-success bound gives
$\mathbb E[p_{S_j}\mid S_{j-1}]\le a p_{S_{j-1}}$; otherwise
monotonicity gives $p_{S_j}\le a^k$. Hence $F_j\le aF_{j-1}$.
Since $F_0=1$ and $S_k$ is a uniform $k$-subset,
\begin{align*}
 \mathbb E_S p_S
 &\le a^k+F_k
 &&\text{(split at the cutoff $a^k$)}\\
 &\le a^k+a^kF_0=2a^k
 &&\text{(iterate $F_j\le aF_{j-1}$).}
\end{align*}
The subsequent counting argument gives
\begin{align*}
 \Pr[N\ge\lceil(v+\delta)n\rceil]
 &\le2\left(\frac{a}{a+\delta/2}\right)^k
 &&\text{(subset counting)}\\
 &\le2\exp\!\left[-\frac{\delta k}{2(a+\delta/2)}\right]
 &&\text{($\ln(1-u)\le-u$).}
\end{align*}
Indeed, on the threshold event a uniform $k$-subset is all winning with
probability at least $(v+\delta-k/n)^k\ge(a+\delta/2)^k$.
Since $k\ge t_{\rm th} n/(2D_0)$ when $k\ge1$, and
\begin{align*}
 (a+\delta/2)D_0
 &\le L+\tfrac14+\frac{3}{e\ln2}
 &&\text{($a+\delta/2\le\min\{1,3a\}$)}\\
 &<3L
 &&\text{($L\ge1$),}
\end{align*}
the stated exponent follows. Here we used
$a\log(1/a)\le1/(e\ln2)$ and $ t_{\rm th}\le1/4$.
If $k=0$, then
\begin{align*}
  t_{\rm th}\delta n/L&<\delta D_0/L
 &&\text{($ t_{\rm th} n<D_0$)}\\
 &\le13/4
 &&\text{($\delta\lambda\le\delta\log(4/\delta)\le2$).}
\end{align*}
The stated right-hand side then exceeds one. This accounts for rounding
at every $n$.

\parrepbibliography
\end{document}